\documentclass[runningheads]{llncs}
\usepackage[utf8]{inputenc}        
\usepackage[T1]{fontenc}             
\usepackage[USenglish]{babel}        
\usepackage[intlimits]{amsmath}      
            
\usepackage{graphicx}
\usepackage{setspace}
\usepackage{comment}
\usepackage{standalone}
\usepackage{subcaption}

\usepackage{tikz}
\usetikzlibrary{graphs,calc,decorations.markings}
\usepackage{todonotes}
\usepackage{hyperref}

\usepackage{algorithm}
\usepackage[noend]{algpseudocode}

\newcommand{\NP}{\textsf{NP}}
\newcommand{\APX}{\textsf{APX}}
\newcommand{\FPT}{\textsf{FPT}}

\newcommand{\BCEVS}{\textsf{BCEVS}}
\newcommand{\BCEOVS}{\textsf{BCEOVS}}

\newcommand{\bigo}{\mathcal{O}}

\newtheorem{observation}{Observation}
\newtheorem{construction}{Construction}

\begin{document}

\title{Bicluster Editing with Overlaps: A Vertex Splitting Approach\thanks{A preliminary version of this paper has been presented at the $36^{th}$ International Workshop on Combinatorial Algorithms (IWOCA 2025) (see \cite{AbuKhzamIM25}).}\thanks{This research project was supported by the Lebanese American University under the President’s Intramural Research Fund PIRF0056.}}

\author{Faisal N. Abu-Khzam, Lucas Isenmann and Zeina Merchad}
\authorrunning{Abu-Khzam, Isenmann, and Merchad} 
\institute{
Department of Computer Science and Mathematics\\
Lebanese American University\\ 
Beirut, Lebanon.\\
}
\maketitle

\begin{abstract}
The {\sc BiCluster Editing} problem aims at editing a given bipartite graph into a disjoint union of bicliques via a minimum number of edge deletion or addition operations. As a graph-based model for data clustering, the problem aims at a partition of the input dataset, which cannot always obtain meaningful clusters when some data elements are expected to belong to more than one cluster each. 
To address this limitation, we introduce the {\sc Bicluster Editing with Vertex Splitting} problem (\BCEVS{}) which consists of finding a minimum sequence of edge editions and vertex splittings such that the resulting graph is a disjoint union of bicliques. The vertex splitting operation consists of replacing a vertex $v$ with two vertices whose union of neighborhoods is the neighborhood of $v$. 
We also introduce the problem of {\sc Bicluster Editing with One-Sided Vertex Splitting} (\BCEOVS{}) where we restrict the splitting operations to the only one set of the two sets forming the bipartition. 
We prove that the two problems are $\NP$-complete even when restricted to bipartite planar graphs of maximum degree three.
Moreover, assuming the {\sc Exponential Time Hypothesis} holds, there is no $2^{o(n)}n^{O(1)}$-time (resp. $2^{o(\sqrt{n})}n^{O(1)}$-time) algorithm for \BCEVS{} and \BCEOVS{} on bipartite (resp. planar) graphs with maximum degree three, where $n$ is the number of vertices of the graph. 
Furthermore we prove both problems are \textsc{APX}-hard and solvable in polynomial time on trees.
On the other hand, we prove that \BCEOVS{} is fixed parameter tractable with respect to solution size by showing that it admits a polynomial size kernel.


\end{abstract}

\section{Introduction}

Cluster Editing is a classical problem with numerous applications across fields, most importantly in computational biology and gene expressions \cite{rahmann2007exact,tan2020gene}.
The problem, as originally proposed, involves determining whether a given graph $G$ can be transformed into a graph consisting of a disjoint union of cliques through the addition or deletion of at most $k$ edges, where $k$ is a given parameter.
Cluster Editing was shown to be \NP-hard \cite{KM86,Krivanek1986np,shamir2004cluster}, and multiple parameterized and approximation algorithms have been developed to address it~\cite{gramm2005graph,Bocker12,ChenMeng,DAddario2014,fadiel2006computational,fellows2007efficient,guo2009more}.

When the input and output graphs are restricted to bipartite graphs, the problem is referred to as {\sc Bicluster Editing}. The main question is whether it is possible to modify (by deleting or adding) at most $k$ edges in the input bipartite graph so that the resulting graph becomes a disjoint union of bicliques, again $k$ is a given parameter bounding the number of allowed edge modifications. 
This problem is greatly used in the field of computational biology and in the analysis of gene expressions ~\cite{Madeira2004}, among other things.
{\sc Bicluster Editing} has been shown to be \NP-complete~\cite{amit2004bicluster-thesis}, and several parameterized and approximation algorithms have been proposed. 
The simplest fixed-parameter algorithm of Bi-cluster Editing runs in $O(4^k +|E|)$, simply by exhaustively trying all possible editing operations when an induced path of length three is found~\cite{ProttiSS09}. Guo \textit{et al.}~\cite{guo2007} improved the running time to $O(3.24^k +|E|)$.
More recently, Xiao and Kou improved the running time bound to $O^*(2.9312^k)$~\cite{Xiao2022} while Tsur proposed a branching algorithm with a running time of $O^*(2.636^k)$~\cite{Tsur2021} and further improved it to run in $O^*(2.22^k)$ \cite{Tsur2023}.

Cluster Editing with Overlapping Communities extends the traditional Cluster Editing problem by allowing vertices to belong to more than one cluster, or {\em to split} among them. 
This particular vertex splitting operation permits a vertex $v$ to be replaced by two vertices whose combined neighborhoods is the neighborhood of $v$.  
Consequently, the data element represented by $v$ can simultaneously belong to more than one cluster, possibly by splitting $v$ one or more times, which is an obvious practical objective. 
Vertex splitting has been introduced and studied in a number of recent articles ~\cite{abu2018,abukhzam2023clustereditingvertexsplitting,AbuKhzamDIT25}. In this paper, we use this operation for the first time in the realm of biclustering and study the complexity of the corresponding problems.

Biclustering is often used when processing raw data given in some tabular form, where rows correspond to data elements and columns correspond to features or some other form of attributes. This table is viewed as the incidence matrix of a bipartite graph after applying a thresholding technique to ``binarize'' the table entries. 
Typically, clustering algorithms aim at grouping the data elements (the rows). By introducing vertex splitting as another (optional) editing operation, we consider the case where splitting is allowed on both sides of the bipartite graph, as well as the case where it is allowed only on one side. This depends on the user's objective. For example, if data elements are not allowed to belong to more than one cluster then it is possible that the features can be common to clusters.
The two corresponding problems are {\sc Bicluster Editing with Vertex Splitting} (BCEVS) and {\sc Bicluster Editing with One Sided Vertex Splitting} (BCEOVS).
In general, this approach can uncover significant relationships that conventional biclustering techniques might overlook.

\textbf{Our contribution.}  
We prove that both BCEVS and BCEOVS are $\NP$-complete, even when restricted to planar graphs of maximum degree three. In addition, and assuming the {\sc Exponential Time Hypothesis} holds, we prove that there is no $2^{o(n)}n^{O(1)}$-time algorithm for \BCEVS{} and \BCEOVS{} on bipartite graphs with maximum degree three. We also show that there is no $2^{o(\sqrt{n})}n^{O(1)}$-time algorithm when the input is restricted to planar graphs (again, modulo the ETH).  
Furthermore, we prove that both problems are \textsc{APX}-hard. 
On the positive side, we show that \BCEOVS{} is fixed-parameter tractable and admits a polynomial size kernel, and we also show that the two problems are solvable in polynomial time on trees.

\section{Preliminaries}

We adopt the following common graph-theoretic terminology. A graph $G = (V, E)$ is said to be \textit{bipartite} if its vertex set $V$ can be divided into two disjoint sets $A$ and $B$ such that every edge $e \in E$ connects a vertex in $A$ to a vertex in $B$. That is, there are no edges between vertices within the same subset.
If $G$ is a bipartite graph consisting of subsets $A$ and $B$, then a biclique in $G$ is defined by two subsets $A' \subseteq A$ and $B' \subseteq B$ such that every vertex in $A'$ is connected to every vertex in $B'$.
The \textit{open neighborhood} \(N(v)\) of a vertex \(v\) is the set of vertices adjacent to it. The degree of \(v\) is the number of edges incident on \(v\), which is \(|N(v)|\) since we only consider simple graphs.
A \textit{geodesic path}, or shortest path, between two vertices in a graph is a path that has the smallest number of edges among all paths connecting these two vertices.
We denote by $G[X]$ the subgraph of a graph $G=(V,E)$ that is \emph{induced} by the vertices in $X\subset V$. In other words, 
$G[X]$ is formed from $G$ by taking a subset $X$ of its vertices and all of the edges in $G$ that have both endpoints in $X$.

We consider the following operations on a bipartite graph $G= (A,B,E)$: edge addition, edge deletion and vertex splitting. 
Vertex splitting is an operation that replaces a vertex \(v\) by two copies \(v_1\) and \(v_2\) such that \(N(v) = N(v_1) \cup N(v_2)\). The operation results in a new bipartite graph. An \textit{exclusive} vertex split requires that \(N(v_1) \cap N(v_2) = \emptyset\). In this paper, we do not assume a split is exclusive, but our proofs apply to this restricted version.
We define the following new problems.

\vspace{7pt}

\noindent
{\sc Bicluster Editing with Vertex Splitting}

\noindent
\underline{Given:} A bipartite graph \( G = (A,B, E) \), along with positive integer \( k \).

\noindent
\underline{Question:} Can we transform \( G \) into a disjoint union of  bicliques by performing at most \( k \) edits including vertex splitting operations?

\vspace{7pt}

\noindent
{\sc Bicluster Editing with One-Sided Vertex Splitting}

\noindent
\underline{Given:} A graph \( G = (A,B, E) \), along with positive integer  \( k \).

\noindent
\underline{Question:} Can we transform \( G \) into a disjoint union of bicliques by performing at most \( k \) edits where only vertices from the set $B$ are subjected to vertex splitting operations?

\vspace{7pt}

A geodesic of length $3$ in a bipartite graph consists of 4 vertices that are not part of a biclique (simply because the endpoints are non-adjacent). Such a geodesic is treated as a forbidden structure, or conflict, that is 
to be {\em resolved} in order to obtain a disjoint union of bicliques.
To explicate, we say that an operation resolves a geodesic of length 3 in a bipartite graph $G$ if, after applying this operation, the two end vertices of the geodesic are at distance different from $3$.
We say that two geodesics of length $3$ are \textit{independent} if there does not exist an operation which resolves simultaneously both of them.

\begin{observation}
There are six types of operations that can resolve a geodesic $(a,b,c,d)$: deleting one of the edges $ab$, $bc$, $cd$, or adding the edge $ad$ or splitting the vertex $b$ (resp. $c$) such that one copy is adjacent to $a$ (resp. $b$) and not to $c$ (resp. $d$) and the other copy is adjacent to $c$ (resp. $d$) and not to $a$ (resp. $b$). 
Furthermore, there are several ways to split vertex $b$ or $c$ to resolve the geodesic if these vertices are adjacent to other vertices.
\end{observation}

For a bipartite graph $G=(A,B,E)$, we denote by $bceovs(G,A)$ the minimum length of a sequence of edge editions on $G$ and vertex splittings on $A$ turning $G$ into a disjoint union of bicliques.
On the other hand, 
we denote by $bcevs(G)$ the minimum length of a sequence of edge editions and vertex splittings (applied to vertices of $A$ or $B$) on $G$ to turn $G$ into a disjoint union of bicliques.

\subsection*{Relation between  $bcevs$ and $bceovs$}

Observe that $bcevs(G) \leq bceovs(G,A)$ for every bipartite graph $G=(A,B,E)$.
This inequality is tight, for example because of the path graph with 3 vertices which has $bcevs$ and $bceovs$ numbers equal to $1$ both.
Furthermore, these two parameters are different because there exists bipartite graphs such that $bcevs(G) < bceovs(G,A)$, as shown in Figure \ref{strictly} below.

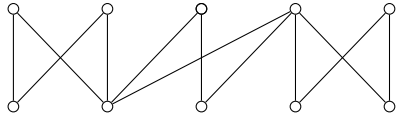
\begin{figure}[hbt!]
    \centering
 \begin{tikzpicture}
    [
        yscale=-1,
        node_style/.style={circle, draw, inner sep=0.05cm} 
    ]
    \def\fscale{5}
	 \node[node_style] (10) at ($\fscale*(-0.25, 0.13)$) {};
	 \node[node_style] (11) at ($\fscale*(-0.5, 0.13)$) {};
	 \node[node_style] (12) at ($\fscale*(-0.5, -0.13)$) {};
	 \node[node_style] (13) at ($\fscale*(-0.25, -0.13)$) {};
	 \node[node_style] (14) at ($\fscale*(0, -0.13)$) {};
	 \node[node_style] (15) at ($\fscale*(0.25, -0.13)$) {};
	 \node[node_style] (16) at ($\fscale*(0, 0.13)$) {};
	 \node[node_style] (17) at ($\fscale*(0, -0.13)$) {};
	 \node[node_style] (18) at ($\fscale*(0.5, -0.13)$) {};
	 \node[node_style] (19) at ($\fscale*(0.25, 0.13)$) {};
	 \node[node_style] (20) at ($\fscale*(0.5, 0.13)$) {};

	 \draw (11) to  (12);
	 \draw (12) to  (10);
	 \draw (10) to  (13);
	 \draw (13) to  (11);
	 \draw (10) to  (14);
	 \draw (10) to  (15);
	 \draw (17) to  (16);
	 \draw (16) to  (15);
	 \draw (15) to  (19);
	 \draw (15) to  (20);
	 \draw (18) to  (19);
	 \draw (18) to  (20);
\end{tikzpicture}
\vspace{10pt}

\caption{Example of a bipartite graph such that $bcevs(G) < bceovs(G,A)$. Here we have $bcevs(G) = 2$ and $bceovs(G,A) = 3$ because in the first case we can split two vertices, one from each side, and in the second case we can split one vertex and delete two edges.}
\label{strictly}
\end{figure}

\section{Complexity of \BCEVS{} and \BCEOVS{}}

The objective is to exhibit a reduction to \textsc{Bicluster Editing with Vertex Splitting} and \textsc{Bicluster Editing with One-sided Vertex Splitting} from a variant of {\sc 3SAT}.
For a 3-CNF formula $F$, we denote by $m$ the number of clauses of $F$.

\begin{construction} \label{construction:reduction}
Consider a 3-CNF formula $F$.
For every variable $v$, we denote by $d(v)$ the number of clauses where $v$ appears and we denote by $c(v)_1, \ldots, c(v)_{d(v)}$ the clauses where $v$ appears.
    
We define a \emph{variable-clause} vertex $v_c$ for every clause $c$ where a variable $v$ appears. Let $j$ be the index of $c$ in the list, defined above, of the clauses where $v$ appears, we define $v_c$ as $v_{6(j-1)+1}$ (resp. $v_{6(j-1)+3}$) if $v$ appears positively (resp. negatively).
We create a graph $G_F$ as follows:
        
\begin{itemize}
        
\item For each variable $v$, we create a cycle $v_1, \ldots, v_{6d(v)}$. We consider the indices modulo $6d(v)$ (so for example $v_{6d(v)+1} = v_1$).
        
\item For each clause $c$, we create a vertex $c$ (we identify a clause and its vertex). 

\item For each clause $c$, containing the variables $u,v,w$, we add the edges $c u_c$, $c v_c$ and $c w_c$ (see the above definition of a variable-clause vertex).

\end{itemize}

\end{construction}

The following Lemma follows immediately from the above construction.

\begin{lemma}
Given a 3-CNF formula $F$ with clauses set $C$, the graph $G_F$ has $19m$ vertices where $m$ is the number of clauses of $F$.
We define $A$ as the vertices of the variables that have an even index.
We define $B$ as the vertices of the variables having an odd index and the vertices of the clauses.
Then the obtained graph $G_F$ is bipartite with bipartition $(A,B)$.
Furthermore, the graph has maximum degree $3$.
\end{lemma}

The cycles created in the above construction play a central role in our proof. The main idea is that ``optimally'' transforming an even-length cycle into a disjoint union of bicliques requires only edge-deletion operations.

\begin{lemma}\label{lemma:cycle}
A cycle of length $6k$, with $k \geq 2$, requires at least $2k$ operations to be turned into a disjoint union of bicliques.
The solution sequences of operations of length $2k$ are the following three ones: delete the edges $v_{1+3i} v_{2+3i}$ for every $i$, delete the edges $v_{2+3i} v_{3+3i}$ for every $i$ and delete the edges $v_{3+3i} v_{4+3i}$ for every $i$.
\end{lemma}

\begin{proof}

We denote by $v_1, \ldots, v_{6k}$ the consecutive vertices of the cycle in question. For $k \geq 2$, we consider the following geodesics:
For every $i \in \{0, \ldots, k-1\}$, we consider the geodesic $v_{4i + 1}, v_{4i+2}, v_{4i+3}, v_{4i+4}$.
These geodesics do not share edges and inner vertices, and no pair of geodesics have the same end-vertices.
Therefore we need at least one operation to ``resolve'' each of them, which means we need at least $2k$ operations to turn $C_{6k}$ into a disjoint union of bicliques.
    
Consider a sequence of $2k$ operations turning $C_{6k}$ into a union of bicliques. Let us prove that there is only three possible sequences.
Assume that there is no edge deletion in the sequence. Then the geodesics of length 4: $v_{2i}, v_{2i+1}, v_{2i+2}, v_{2i+3}$ for every $i \in \{1, \ldots, 3k\}$ are such that no edge addition and no vertex splitting on the graph can solve two conflicts simultaneously.
We deduce that we need at least $3k$ operations in this case, a contradiction with the assumed length of the sequence.
Thus the sequence of operations has at least one edge deletion.

Assume that there exists $i$ and $j >i$ such that $v_i v_{i+1}$ and $v_{j}v_{j+1}$ are deleted.
We shall now prove that $i=j$ modulo 3. Otherwise we can find $2(k-1)+1$ independent geodesics of length $4$ each on the subpath $v_{i+1}, v_{i+2}, \ldots, v_{j}$ and on the subpath $v_{j+1}, v_{j+2}, \ldots, v_i$. An example is provided in Figure~\ref{fig:C12_different_classes}.
Therefore this sequence would be of length at least $2(k-1)+1+2$, a contradiction.
We deduce that edge deletions only occurs with same index modulo 3.

\begin{figure}[htb!]
    \centering
\begin{tikzpicture}
    [
        yscale=-1,
        node_style/.style={circle, draw, inner sep=0.05cm},
        deleted/.style={postaction=decorate, decoration={markings, mark=at position 0.47 with {\arrow{|}}, mark=at position 0.53 with {\arrow{|}}  } }
    ]
    \def\fscale{3} 
	 \node[node_style] (0) at ($\fscale*(0.5, 0)$) {};
	 \node[node_style] (1) at ($\fscale*(0.43, 0.25)$) {};
	 \node[node_style] (2) at ($\fscale*(0.25, 0.43)$) {};
	 \node[node_style] (3) at ($\fscale*(0, 0.5)$) {};
	 \node[node_style] (4) at ($\fscale*(-0.25, 0.43)$) {};
	 \node[node_style] (5) at ($\fscale*(-0.43, 0.25)$) {};
	 \node[node_style] (6) at ($\fscale*(-0.5, 0)$) {};
	 \node[node_style] (7) at ($\fscale*(-0.43, -0.25)$) {};
	 \node[node_style] (8) at ($\fscale*(-0.25, -0.43)$) {};
	 \node[node_style] (9) at ($\fscale*(0, -0.5)$) {};
	 \node[node_style] (10) at ($\fscale*(0.25, -0.43)$) {};
	 \node[node_style] (11) at ($\fscale*(0.43, -0.25)$) {};


	 \draw[color=green] (0) to  (1);
	 \draw[color=green] (1) to  (2);
	 \draw[color=green] (2) to  (3);
	 \draw (3) to  (4);
	 \draw[color=green] (4) to  (5);
	 \draw[color=green] (5) to  (6);
	 \draw[color=green] (6) to  (7);
 \draw[deleted] (7) to  (8);
 \draw[color=green] (8) to  (9);
	 \draw[color=green] (9) to  (10);
	 \draw[color=green] (10) to  (11);
	 \draw[deleted] (11) to  (0);

\end{tikzpicture}\\
\vspace{10pt}

\caption{Example in the cycle $C_{12}$ where two edge deletions occur at indices that are not equal modulo 3. As there are 3 independent geodesics of length 4 (in green), we need at least 5 operations in this case.}
    \label{fig:C12_different_classes}
\end{figure}
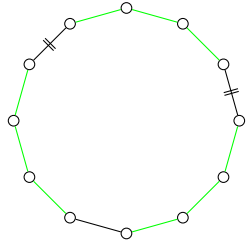

As the sequence is deleting at least one edge, there exists $i$ such that $v_i v_{i+1}$ is deleted.
Because of the previous result, all the other edge deletions of $v_{k} v_{k+1}$ are such that $k = i$ modulo 3.
Suppose that there exists an index $j$ such that $v_{i+3j} v_{i+3j+1}$ is not deleted.
Then we can find $2k$ geodesics of length $4$ such that all these geodesics are independent except for 2 of them which share a the common edge $v_{i+3j}v_{i+3j+1}$.
But as this edge is supposed to not be deleted then these two geodesics are also independent.
We conclude that all the edges $v_{i+3j} v_{i+3j+1}$ are deleted.

\end{proof}

\begin{lemma}\label{lemma:cycle_combinatorics}
Let $A$ be a subset of  $\{0, 1, \ldots, 6k-1\}$ for $k \geq 1$ such that every element of $A$ equals 0 or 2 modulo 6.
Considering an increasing enumeration  $a_1, \ldots, a_p$ of $A$, then for every $i \in [p]$, $a_{i+1}- a_i = 0,2,4$ modulo 6 and there is as much $i \in [p]$ such that $a_{i+1} - a_i = 2$ (modulo 6) as much as $i \in [p]$ such that $a_{i+1}-a_i = 4$ modulo 6. (we consider the indices modulo $p$).
\end{lemma}

\begin{proof}
First observe that if $a_{i+1}-a_i = 0$ (modulo 6) then $a_{i+1} = 0$ and $a_i = 0$ (modulo 6) or $a_{i+1} = 2[6]$ and $a_i = 2$ (modulo 6).
    
We now define $\alpha$ (resp. $\beta$) as the number of $i \in [p]$ such that $a_{i+1}-a_i = 2$ modulo 6  (resp. $=4$ modulo 6).
Assume (for a contradiction) that $\alpha \not= \beta$.
Thus $\alpha > \beta$ or $\beta > \alpha$.
Assume that $\alpha > \beta$. 
Then there exists $i \in [p]$ and $j \in [p]$ such that $a_{i+1}- a_i = 2$ (modulo 6) and $a_{j+1}-a_j = 2$ (modulo 6) and $a_{k+1}-a_k = 0$ (modulo 6) for every $k \in [i+1,j-1]$.
As $a_{i+1}-a_i = 2$ (modulo 6), then $a_i = 0$ and $a_{i+1} = 2$ modulo 6.
By induction, we show that $a_k = 2$ (modulo 6) for every $k \in [i+1, j]$.
As $a_{j+1}-a_j = 2$ (modulo 6), then $a_{j+1} = 2[6]$ and $a_j = 0$ (modulo 6).
This is a contradiction because $a_j =2$ (modulo 6) based on the previous equality.

In the same way, we prove that it is not possible that $\beta > \alpha$.
We conclude that $\alpha =\beta$.

\end{proof}

\begin{lemma}\label{lemma:path_union}
Let a cycle $v_1, \ldots, v_{6k}$ of length $6k$ and a strictly increasing sequence $i_1, \ldots, i_b$ of integers in $[1,6k]$ such that the edges $v_{i_{j-1}} v_{i_j}$ and $v_{i_j} v_{i_{j+1}}$ are deleted for every $j \in [1,b]$.
Suppose that for every $j \in [1,b]$, $i_{j+1} - i_j$ equals 0 or 2 modulo 6.
Then the remaining paths require at least $2k-b$ operations to be turned into a disjoint union of bicliques.
\end{lemma}

\begin{proof}
Because of Lemma~\ref{lemma:cycle_combinatorics}, there is as much paths of length $2$ and 0 modulo 6.

Let $\alpha, \beta$ and $\gamma$ be the number of paths of length respectively $2,4$ and $0$ modulo 6.
Then $\alpha + \beta + \gamma = b$.
As $\alpha = \gamma$, we deduce that $2\alpha + \beta = b$.
    
We denote by $l_i = 6q_i + r_i$ for every path. We define $Q = \sum_i l_i$. Thus
    \[
    \sum_{i} l_i = \sum_{i} 6q_i + 2 \alpha + 4 \beta = 6 Q + 2\alpha + 4 (b-2\alpha) = 6 Q- 6\alpha + 4 b
    \]

The sum of the lengths of the paths equals $6k-2b$. Therefore $6k - 2b  = 6Q - 6 \alpha + 4b$ and $6k - 6b = 6Q - 6 \alpha$ and $k + \alpha = Q + b$.

As a path of length $l$ needs at least $\lfloor \frac{l}{3} \rfloor$ operations to be turned into a disjoint union of bicliques
    \begin{align*}
        \sum_{i} \left\lfloor \frac{l_i}{3} \right\rfloor &=  \sum 2 q_i + \lfloor r_i \rfloor = \sum_{i} 2 q_i + \beta = 2 Q + \beta \\
        &= 2k + 2 \alpha - 2b + \beta = 2k + 2\alpha - 2b +b - 2\alpha = 2k - b \;.
    \end{align*}

\end{proof}

\begin{theorem} \label{theorem:npc}
\BCEOVS{} and \BCEVS{} are \NP-complete even when restricted to bipartite graphs of maximum degree three.
\end{theorem}

\begin{proof}
These problems are clearly in \NP. Let $F$ be a 3-CNF, we denote the set of clauses by $C$ and the set of variables by $V$.
Let $G$ be the bipartite graph obtained by Construction~\ref{construction:reduction}.
We set $k=8m$ where $m$ is the number of clauses. Let us prove that the graph $G$ has a sequence of at most $k$ operations such that it turns $G$ into a disjoint union of bicliques if and only if $F$ is satisfiable.

Assume that $F$ is satisfiable.
For every true (resp. false) variable $v$ we delete the edges $v_{1+3i} v_{2+3i}$ (resp. $v_{2+3i} v_{3+3i}$) for every $i$.
For every clause $c$, there exists a variable $v$ appearing in $c$ which satisfies $c$.
Let $u$ and $w$ be the two other variables appearing in $c$.
We delete the edges $v u_c$ and $v w_c$.

The resulting graph is a union of bicliques because the variable cycles have been turned into disjoint paths of length 2 and for every clause, the remaining edge is connected to the middle of a path of length 2. 
Therefore the connected components are stars with 4 vertices and paths of length 2. We have done $6d(v)/3 = 2d(v)$ deletions for every variable $v$. We have done two edge deletions for every clause.
In total, we have done $2m + \sum_{v \in V} 2d(v) = (2+ 2\cdot3) m = 8m = k$ operations.

    \begin{figure}[htb!]
        \centering
\begin{tikzpicture}
    [
        yscale=-1,
        node_style/.style={circle, draw, inner sep=0.1cm} 
    ]
    \def\fscale{8} 
    
	 \node[node_style] (0) at ($\fscale*(-0.24, 0.01)$) {};
	 \node[node_style] (1) at ($\fscale*(-0.25, 0.07)$) {};
	 \node[node_style] (2) at ($\fscale*(-0.3, 0.12)$) {};
	 \node[node_style, label={above:$a_7$}] (3) at ($\fscale*(-0.37, 0.14)$) {};
	 \node[node_style] (4) at ($\fscale*(-0.43, 0.12)$) {};
	 \node[node_style] (5) at ($\fscale*(-0.48, 0.07)$) {};
	 \node[node_style] (7) at ($\fscale*(-0.48, -0.06)$) {};
	 \node[node_style] (8) at ($\fscale*(-0.43, -0.11)$) {};
	 \node[node_style, label={above:$a_1$}] (9) at ($\fscale*(-0.37, -0.12)$) {};
	 \node[node_style] (10) at ($\fscale*(-0.3, -0.11)$) {};
	 \node[node_style] (11) at ($\fscale*(-0.25, -0.06)$) {};
	 \node[node_style] (12) at ($\fscale*(-0.5, 0.01)$) {};
	 \node[node_style] (6) at ($\fscale*(0.13, 0.01)$) {};
	 \node[node_style] (13) at ($\fscale*(0.11, 0.07)$) {};
	 \node[node_style] (14) at ($\fscale*(0.07, 0.12)$) {};
	 \node[node_style, label={above:$b_7$}] (15) at ($\fscale*(0, 0.14)$) {};
	 \node[node_style] (16) at ($\fscale*(-0.07, 0.12)$) {};
	 \node[node_style] (17) at ($\fscale*(-0.11, 0.07)$) {};
	 \node[node_style] (18) at ($\fscale*(-0.11, -0.06)$) {};
	 \node[node_style] (19) at ($\fscale*(-0.07, -0.11)$) {};
	 \node[node_style, label={below:$b_1$}] (20) at ($\fscale*(0, -0.12)$) {};
	 \node[node_style] (21) at ($\fscale*(0.07, -0.11)$) {};
	 \node[node_style] (22) at ($\fscale*(0.11, -0.06)$) {};
	 \node[node_style] (23) at ($\fscale*(-0.13, 0.01)$) {};
	 \node[node_style] (24) at ($\fscale*(0.5, 0.01)$) {};
	 \node[node_style] (25) at ($\fscale*(0.48, 0.07)$) {};
	 \node[node_style] (26) at ($\fscale*(0.43, 0.12)$) {};
	 \node[node_style] (27) at ($\fscale*(0.37, 0.14)$) {};
	 \node[node_style] (28) at ($\fscale*(0.3, 0.12)$) {};
	 \node[node_style, label={right:$c_9$}] (29) at ($\fscale*(0.25, 0.07)$) {};
	 \node[node_style] (30) at ($\fscale*(0.25, -0.06)$) {};
	 \node[node_style] (31) at ($\fscale*(0.3, -0.11)$) {};
	 \node[node_style, label={above:$c_1$}] (32) at ($\fscale*(0.37, -0.12)$) {};
	 \node[node_style] (33) at ($\fscale*(0.43, -0.11)$) {};
	 \node[node_style] (34) at ($\fscale*(0.48, -0.06)$) {};
	 \node[node_style] (35) at ($\fscale*(0.24, 0.01)$) {};
	 \node[node_style, label={above:$a \vee b \vee c$}] (36) at ($\fscale*(0, -0.24)$) {};
	 \node[node_style, label={below:$a \vee b \vee \overline{c}$}] (37) at ($\fscale*(0, 0.24)$) {};

	 \draw (12) to  (7);
	 \draw (7) to  (8);
	 \draw (8) to  (9);
	 \draw (9) to  (10);
	 \draw (10) to  (11);
	 \draw (11) to  (0);
	 \draw (0) to  (1);
	 \draw (1) to  (2);
	 \draw (2) to  (3);
	 \draw (3) to  (4);
	 \draw (4) to  (5);
	 \draw (5) to  (12);
	 \draw (23) to  (18);
	 \draw (18) to  (19);
	 \draw (19) to  (20);
	 \draw (20) to  (21);
	 \draw (21) to  (22);
	 \draw (22) to  (6);
	 \draw (6) to  (13);
	 \draw (13) to  (14);
	 \draw (14) to  (15);
	 \draw (15) to  (16);
	 \draw (16) to  (17);
	 \draw (17) to  (23);
	 \draw (35) to  (30);
	 \draw (30) to  (31);
	 \draw (31) to  (32);
	 \draw (32) to  (33);
	 \draw (33) to  (34);
	 \draw (34) to  (24);
	 \draw (24) to  (25);
	 \draw (25) to  (26);
	 \draw (26) to  (27);
	 \draw (27) to  (28);
	 \draw (28) to  (29);
	 \draw (29) to  (35);
	 \draw (9) to  (36);
	 \draw (36) to  (20);
	 \draw (36) to  (32);
	 \draw (3) to  (37);
	 \draw (37) to  (29);
	 \draw (37) to  (15);
\end{tikzpicture}

\vspace{10pt}

\caption{The graph constructed by Construction~\ref{construction:reduction} for the 3-CNF $(a \vee b \vee c) \land (a \vee b \vee \overline{c})$.}
\label{fig:enter-label}
\end{figure}
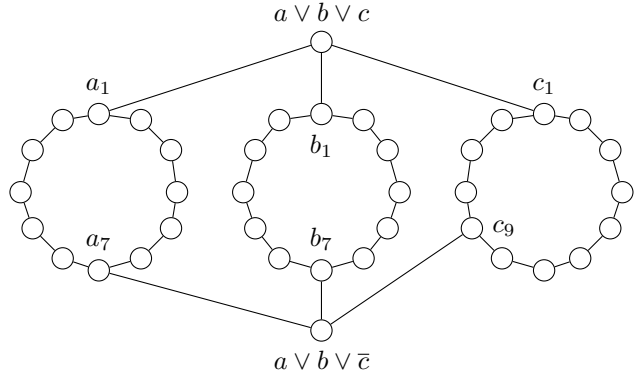

Assume that $G$ can be turned into a disjoint union of bicliques with a sequence of at most $k$ operations. For every variable $v$, we denote by $Op(v)$ the set of edge editions  and vertex splittings done between vertices of the cycle of $v$. For every clause $c$, we denote by $Op(c)$ the set of the splits of $c$ and edge deletions which are incident to $c$ and edge additions between $c$ and the vertices $u_c, u_{c-1}, u_{c+1}, v_c, v_{c-1}, v_{c+1}, w_c, w_{c-1}, w_{c+1}$ where $u,v$ and $w$ are the variables in $c$.
These sets are pair-wise disjoint.

For each variable $v$, according to Lemma~\ref{lemma:cycle}, $Op(v)$ is of size at least $2d(v)$.
    
We denote by $a_0$ (resp. $a_1$) the number of clause $c$ such that $Op(c)$ is of size $0$ (resp. $1$).
We denote by $a_2$ the number of clauses such that $Op(c)$ is of size at least $2$.
Thus $m = a_0 + a_1 + a_2$.
Let us prove that $a_0 =0$ and $a_1 = 0$.

    Let $c$ be a clause such that $Op(c)$ is of size 0.
    In the \BCEOVS{} problem the edges $u_{c-1} u_c$ and $u_c u_{c+1}$ must be deleted for each variable $u$ appearing in $c$.
    Suppose there exists $u$ such that it is not the case, then there still exist the geodesic $u_{c\pm1}, u_c, c, v_c$ where $v$ is another variable appearing in $c$.
    In the \BCEVS{} problem, it is also possible to solve the problems of the geodesics $u_{c\pm1}, u_c, c, v_c$ by splitting the vertex $v$.
    So there are three cases for each variable $u$ appearing in $c$.
    Either both edges $u_{c-1} u_c$ and $u_c u_{c+1}$  are deleted, either one edge among $u_{c-1} u_c$ and $u_c u_{c+1}$ is deleted and $u$ is split, either only $u$ is split.

    Let $c$ be a clause such that $Op(c)$ is of size 1.
    In the \BCEOVS{} problem there exists two variables $u$ and $v$ appearing in $c$ such that the edges $u_{c-1} u_c$, $u_c u_{c+1}$ and $v_{c-1} v_c$, $v_c v_{c+1}$ are deleted.
    Indeed, if $c$ is split or if an edge incident to $c$ is deleted, then there remains two variables $u$ and $v$ appearing in $c$ such that the edges $c u_c$ and $c v_c$ are still present.
    We deduce that the four previous edges must be deleted.
    Otherwise there is an edge addition in $Op(c)$ which occurs between vertices of $u$ and of $v$.
    For geodesics where the third variable \( w \) appears in \( c \), the edges \( u_{c-1} u_c \), \( u_c u_{c+1} \), as well as \( v_{c-1} v_c \) and \( v_c v_{c+1} \), should be removed.
    In the \BCEVS{} problem, there exists two vertices appearing in $c$ such that either both edges $u_{c-1} u_c$ and $u_c u_{c+1}$  are deleted, either one edge among $u_{c-1} u_c$ and $u_c u_{c+1}$ is deleted and $u$ is split, either only $u$ is split.

    For every variable $v$, we denote by $B(v)$ the clauses $c$ such that $v$ is appearing in $c$ and such that $|Op(c)| \leq 1$ and such that the two edges $v_{c-1} v_c$ and $v_c v_{c+1}$ are deleted.
    We denote by $b(v)$ the size of $B(v)$ and by $r(v)$ the number of clauses $c$ such that $v$ is appearing in $c$ and such that $|Op(c)| \leq 1$ and such that $v_c$ is split and such that the two edges $v_{c-1} v_c$ and $v_c v_{c+1}$ are deleted.
    Thus $3 a_0 + 2 a_1 = \sum_{v \in V} b(v) + r(v)$ by double counting.
    
    Let $v$ be a variable.
    For every $c \in B(v)$, the edges $v_{c-1} v_c$ and $v_c v_{c+1}$ are deleted.
 All the indices of $B(v)$ are 0, 2 modulo 6 by definition of the graph.
    In reference to Lemma~\ref{lemma:path_union}, $Op(v)$ is of size at least $2d(v) +  b(v)$.
 The total number of operations is at least
 
    \begin{align*}
       &\geq \sum_{c \in C} |Op(c)| + \sum_{v\in V} |Op(v)| \\
        &\geq 2a_2 + a_1 + \sum_{v \in V} (2d(v) +b(v) + r(v)) \geq 2a_2 + a_1 + \sum_{v \in V} 2d(v) +  \sum_{v \in V} (b(v) + r(v))\\
        &\geq 2a_2 + a_1 + 6m + (3 a_0 + 2a_1) \geq 2m + a_0 + a_1 + \sum_{v \in V} 2d(v) \geq k + a_0 + a_1 \;. \\
    \end{align*}

    \begin{figure}[!htb]
        \centering
\begin{tikzpicture}
    [
        yscale=-1,
        node_style/.style={circle, draw, inner sep=0.1cm} 
    ]
    \def\fscale{7} 
	 \node[node_style, label={left:$v_{c-1}$}] (0) at ($\fscale*(-0.38, -0.17)$) {};
	 \node[node_style] (1) at ($\fscale*(-0.48, -0.33)$) {};
	 \node[node_style] (2) at ($\fscale*(-0.38, -0.5)$) {};
	 \node[node_style] (4) at ($\fscale*(-0.19, -0.5)$) {};
	 \node[node_style, label={right:$v_{c+1}$}] (5) at ($\fscale*(-0.1, -0.33)$) {};
	 \node[node_style, label={above:$v_c$}] (6) at ($\fscale*(-0.19, -0.17)$) {};
	 \node[node_style, label={below:$c$}] (3) at ($\fscale*(-0.1, 0)$) {};
	 \node[node_style, label={right:$u_c$}] (7) at ($\fscale*(-0.19, 0.17)$) {};
	 \node[node_style, label={right:$w_c$}] (9) at ($\fscale*(0.1, 0)$) {};
	 \node[node_style] (8) at ($\fscale*(0.19, -0.17)$) {};
	 \node[node_style] (10) at ($\fscale*(-0.38, 0.17)$) {};
	 \node[node_style] (11) at ($\fscale*(-0.1, 0.33)$) {};
	 \node[node_style] (12) at ($\fscale*(0.19, 0.17)$) {};
	 \node[node_style] (13) at ($\fscale*(-0.48, 0.33)$) {};
	 \node[node_style] (14) at ($\fscale*(-0.38, 0.5)$) {};
	 \node[node_style] (15) at ($\fscale*(-0.19, 0.5)$) {};
	 \node[node_style] (16) at ($\fscale*(0.38, 0.17)$) {};
	 \node[node_style] (17) at ($\fscale*(0.48, 0)$) {};
	 \node[node_style] (18) at ($\fscale*(0.38, -0.17)$) {};

	 \draw (0) to  (1);
	 \draw[dashed] (1) to  (2);
	 \draw[dashed] (2) to  (4);
	 \draw[color=green] (5) to  (6);
	 \draw[color=green] (6) to  (0);
	 \draw (5) to  (4);
	 \draw[color=green] (6) to  (3);
	 \draw[color=green] (3) to  (7);
	 \draw[color=green] (3) to  (9);
	 \draw (9) to  (8);
	 \draw (7) to  (10);
	 \draw (7) to  (11);
	 \draw (9) to  (12);
	 \draw (10) to  (13);
	 \draw (13)[dashed] to  (14);
	 \draw (14)[dashed] to  (15);
	 \draw (15) to  (11);
	 \draw (12) to  (16);
	 \draw[dashed] (16) to  (17);
	 \draw[dashed] (17) to  (18);
	 \draw (18) to  (8);
	 \draw[color=gray] (7) to  (0);
	 \draw[color=gray] (9) to  (5);
	 \draw[color=gray] (7) to  (5);
	 \draw[color=gray] (9) to  (0);
\end{tikzpicture}

\vspace{10pt}

\caption{Example of a clause $c$ where the three green edges incident to $c$ are not deleted and the three green edges incident to $v_c$ are not deleted. The edges in gray should be added or the vertex $v_c$ should be split.}
\label{fig:no_deletion}
\end{figure}
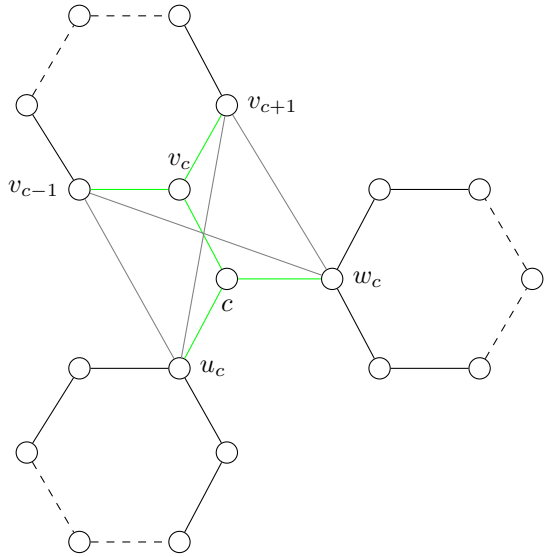

As the sequence is of length at most $k$, we deduce that $a_0 = a_1 = 0$. Thus $|Op(c)| \geq 2$ for every clause $c$. 
We conclude that $|Op(c)| = 2$ for every clause $c$ and $|Op(v)| = 2d(v)$ for every variable $v$. 
In reference to Lemma~\ref{lemma:cycle}, either we delete the edges $v_{1+3k} v_{2+3k}$ for every $k$, either we delete the edges $v_{2+3k} v_{3+3k}$ for every $k$, either we delete the edges $v_{3+3k} v_{4+3k}$ for every $k$.
In the second case, we assign $v$ to positive. Otherwise we assign $v$ to negative.

Let us prove that this assignment satisfies all the clauses. Let $c$ be a clause. If no variable appearing in $c$ is positive, then all these variables are either negative or undefined.
For every variable $v$ appearing in $c$, $v_c$ is connected to a path of length $2$ in the variable cycle. 
Then we need at least 3 operations to solve the 3 conflicts $c, v_c, v_{c+1}, v_{c+2}$ (or $c, v_c, v_{c-1}, v_{c-2}$) for each variable $v$ appearing in $c$.
This contradicts the fact that $Op(c)$ is of size at most $2$. We deduce that there exists a variable $v$ appearing in $c$ which has a positive assignment. 
Therefore, this assignment satisfies all the clauses.
We conclude that the problems are NP-complete.
\end{proof}

\begin{corollary} \label{corollary:planar}
    \BCEVS{} and \BCEOVS{} remain \NP-Complete on bipartite planar graphs with maximum degree three.
\end{corollary}
\begin{proof}
    Consider an instance of \textsc{3SAT-Planar}. 
    The incidence graph of this instance is planar. 
    The graph produced by the previous construction can be also constructed by replacing every variable vertex by a cycle of a certain length with one vertex adjacent to one clause containing the variable.
    Each of these elementary operations conserves the planarity of the graph. 
    We deduce that the produced graph is planar and that the previous construction gives a reduction from \textsc{3SAT-Planar} to \BCEVS{} and \BCEOVS{} restricted to bipartite planar graphs with maximum degree three. 
    The proof is now complete, knowing that \textsc{3SAT-Planar} is NP-complete \cite{David}.
\end{proof}

Since the previous construction is linear in the number of vertices and (resp. Planar) 3-SAT does not admit a $2^{o(n)}n^{O(1)}$ (resp. $2^{o(\sqrt{n})}n^{O(1)})$ time algorithm, unless the Exponential Time Hypothesis (ETH) fails \cite{Cygan2015}. We conclude with the following:

\begin{corollary}
Assuming the ETH holds, there is no $2^{o(n)}n^{O(1)}$ (resp. $2^{o(\sqrt{n})}n^{O(1)}$) time algorithm for \BCEVS{} and \BCEOVS{} on bipartite (resp. planar) graphs with maximum degree three where $n$ is the number of vertices of the graph.
\end{corollary}

\section{\BCEVS{} and \BCEOVS{} on Trees}

In general we have $bcevs(G) \leq bceovs(G,A)$. The idea of our algorithm for computing the $bceovs$ and the $bcevs$ numbers of a tree is to look for a cut vertex separating the graph into a star with at least two vertices and the rest of the tree and to recurse on the subtree.

The first case we investigate is where the cut vertex $y$ is connected to the second subset with only one vertex.

\begin{lemma} \label{lemma:bceovs_rec_simple}
Let $y$ be a cut vertex of a tree $T=(A,B,E)$ partitioning $V(T) \setminus \{y\}$ into $X$ and $Y$ such that $T[X \cup y]$ is a biclique and there exists a vertex at distance two from $y$ in $X$.
If $|N(y) \cap Y| = 1$, then $bceovs(T,A) = 1+ bceovs(T[Y],A)$ and $bcevs(T) = 1 + bcevs(T)$.
\end{lemma}

\begin{proof}
As $|N(y) \cap Y| = 1$, we denote by $z$ the neighbor of $y$ in $Y$.
Let $a$ be a vertex in $X$ at distance 2 from $y$ and $x$ be the vertex of $N(y) \cap N(a)$.
See Figure~\ref{fig:bceovs_tree_simple} for an example.
    
    \begin{figure}[h] 
        \centering
\begin{tikzpicture}
    [
        yscale=-1,
        node_style/.style={circle, draw, inner sep=0.05cm} 
    ]
    \def\fscale{5} 
	 \node[node_style, label={left:$a$}] (0) at ($\fscale*(-0.5, 0)$) {};
	 \node[node_style, label={above:$y$}] (2) at ($\fscale*(0, 0)$) {};
	 \node[node_style, label={above:$z$}] (4) at ($\fscale*(0.25, 0)$) {};
	 \node[node_style, label={above:$x$}] (5) at ($\fscale*(-0.25, 0)$) {};
  \node[node_style] (6) at ($\fscale*(-0.5, -0.25)$) {};
   \node[node_style] (7) at ($\fscale*(-0.5, 0.25)$) {};
	 \node[node_style] (8) at ($\fscale*(0.5, -0.25)$) {};
	 \node[node_style] (9) at ($\fscale*(0.5, 0.25)$) {};
     \node[node_style] (1) at ($\fscale*(0.75, -0.25)$) {};

    \draw (5) to  (0);
    \draw (5) to  (6);
    \draw (5) to  (7);
    
	 \draw (2) to  (4);
	 \draw (2) to  (5);
	 \draw (4) to  (8);
	 \draw (4) to  (9);
      \draw (1) to  (8);
\end{tikzpicture}

\vspace{10pt}
\caption{Example of a tree with a cut vertex $y$ connected to only one vertex $z$ in $Y$ and such that $X$ is a star. 
An optimal solution consists here in deleting the edge $yz$. }
        \label{fig:bceovs_tree_simple}
    \end{figure}
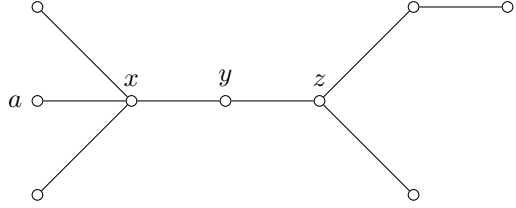
    
    Let $\sigma$ be a sequence of edge deletions and vertex splittings in $T[V-y-X]$.
    We add an edge deletion of $yz$ at the beginning of $\sigma$.
    As this new sequence turns $T$ into a disjoint union of 2-clubs.
    Thus $bcevs(T) \leq bcevs(T[Y])+1$ and $bceovs(T) \leq bceovs(T[Y],A)+1$.

    Let $\sigma$ be a sequence of operations on $G$ turning $G$ into a disjoint union of 2-clubs.
    Because of the geodesic $(a,x,y,z)$, either one of the edges $ax$, $xy$ or $yz$ should be deleted, or add the edge $az$ or split $x$ or $y$.
    We remove all vertex splittings done on $X \cup y$ and all edge editions done on edges incident to $X \cup y$.
    Because of the geodesic $(a,x,y,z)$, at least one operation has been removed.
    We add the edge deletion of $yz$ at the beginning of the sequence $\sigma$.
    This new sequence $\sigma'$ is of length at most the length of $\sigma$ and is still turning $T$ into a disjoint union of 2-clubs.

    Thus, even if restrict the vertex splittings to the vertices of $A$, there exists a minimum sequence deleting the edge $yz$ and doing no operation on edges incident to $X \cup y$ and doing no vertex splittings on $X \cup y$.
    Thus by removing the initial edge deletion $yz$ of $\sigma'$, we get a sequence of operations turning $T[Y]$ into a disjoint union of 2-clubs.
    We deduce that $bcevs(T) \geq bcevs(T[Y]) +1$ and $bceovs(T,A) \geq bceovs(T[Y],A)$.
    We conclude that the announced equalities are true.
\end{proof}

\begin{lemma} \label{lemma:bceovs_rec_multiple}
    Let $y$ be a cut vertex of a tree $T=(A,B,E)$ such that one subset $X$ is a star with at least $2$ vertices.
    If $y$ is connected to at least 2 vertices in $Y$, then 
    \begin{align*}
        bcevs(T)  &= \min( 1 + bcevs(T[Y \cup y]), \; |N(y)\cap Y| + bcevs(T[Y])) \\
        bceovs(T)  &= \min( 1 + bceovs(T[Y \cup y]), \; |N(y)\cap Y| + bceovs(T[Y])) \;.
    \end{align*}
\end{lemma}
\begin{proof}
    In any case, either $y \in A$ or $y \in B$, let us prove that $bceovs(T,A) \leq |N(y)\cap Y| + bceovs(T-y-X,A)$ and that $bceovs(T,A) \leq 1 + bcevs(T[Y],A)$.
    
    Consider a sequence of $T[Y]$ turning this graph into a  disjoint union of bicliques.
    Apply this sequence to $T$ and delete all edges $yz$ where $z \in N(y)\cap Y$.
    Thus this sequence turns $T$ into a disjoint union of bicliques with $|N(y)\cap Y|$ additional operations (edge deletions).
    Thus $bceovs(T,A) \leq |N(y)\cap Y| + bceovs(T[Y],A)$ and $bcevs(T) \leq |N(y)\cap Y| + bcevs(T[Y])$ .

    Consider a sequence of $G[V-X]$ turning this graph into a  disjoint union of bicliques.
    Apply this sequence to $T$ and then delete the edge $xy$.
    Thus this sequence turns $T$ into a disjoint union of bicliques with $|N(y)\cap X|$ additional operations (edge deletions).
    Thus $bceovs(T,A) \leq 1 + bcevs(T-X,A)$ and $bcevs(T) \leq 1 + bcevs(T-X)$.

Consider a sequence of $G$  of length $k$ turning this graph into a disjoint union of bicliques.
If all edges $yz$ with $z \in Y$ are deleted, then we remove all the operations done on vertices of $X$ and the edge additions incident to $y$.
This new sequence of operations still turns the graph $G$ into a disjoint union of bicliques and is of smaller length.
The restriction of this sequence to $G[Y]$ turns this graph into a disjoint union of bicliques.
Furthermore this sequence is of length at most $k - |N(y)\cap Y|$ (because it does not contain the edge deletions $yz$ where $z \in Y$).
We deduce that $bceovs(T,A) \geq |N(y)\cap Y| + bceovs(T[Y],A)$ and $bcevs(T) \geq |N(y)\cap Y| + bcevs(T[Y])$ .

Otherwise there exists $z \in Y$ such that $yz$ has not been deleted.
Because of the geodesic $(a,x,y,z)$ in $G$, either the edges $ax$ or $xy$ are deleted, either the edge $az$ is added, either the vertex $x$ or $y$ is split.
We replace this operation by deleting the edge $xy$ at the beginning of the sequence.
This new sequence has the same length $k$.
We consider the restriction of the sequence to $G[Y \cup \{y\}]$ which has length at most $k-1$ (because it does not contain the edge deletion $xy$).
This new sequence turns the graph into a disjoint union of bicliques.
We deduce that $bceovs(T,A) \geq 1| + bceovs(T[Y\cup y],A)$ and $bcevs(T) \geq 1 + bcevs(T[Y\cup y])$ .

\end{proof}

As the recursive equations are the same for $bcevs$  and $bceovs$ and as if $T$ is a star, then $bcevs(T) = bceovs(T,A) = 0$, we deduce from Lemma~\ref{lemma:bceovs_rec_simple} and Lemma~\ref{lemma:bceovs_rec_multiple} the following Theorem:

\begin{theorem}
    Let $T=(A,B,E)$ be a tree.
    Then $bceovs(T,A) = bcevs(T)$ and there exists an optimal sequence without vertex splitting.
\end{theorem}

\begin{theorem}
   \BCEVS{} and \BCEOVS{} are solvable in polynomial time in trees. 
\end{theorem}

\begin{proof}
    
    Consider a tree $T$ with $n$ vertices and any vertex as the root.
    We consider a postorder numbering from $1$ to $n$ of the vertices such that the deepest branches are visited first.
    
    For every vertex $x$ we denote by $\phi(x)$ the minimum descendant of $x$.
For any $i$ and $j$ such that $j$ is a ancestor of $i$, we define $T[i,j]$ as the induced subgraph of $T$ from the vertices $\{i, i+1, \ldots, j\}$.
Therefore for every subtree $T[i,j]$ where $j$ is an ancestor of $i$, the numbering is still a postorder numbering of the vertices in decreasing depth order.
    
We define $t[i,j] = bcevs(T[i,j])$. Thus $bcevs(T) = t[1,n]$. See Figure~\ref{fig:tree_postorder} for an example.

    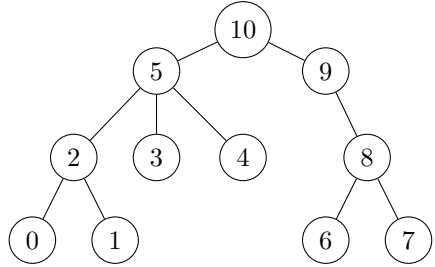
\begin{figure}[h] 
        \centering
\begin{tikzpicture}
    [
        yscale=-1,
        node_style/.style={circle, draw, minimum size=0.5cm} 
    ]
    \def\fscale{5} 
	 \node[node_style] (1) at ($\fscale*(-0.28, 0.28)$) {$1$};
	 \node[node_style] (2) at ($\fscale*(-0.39, 0.06)$) {$2$};
	 \node[node_style] (3) at ($\fscale*(-0.17, 0.06)$) {$3$};
	 \node[node_style] (4) at ($\fscale*(0.06, 0.06)$) {$4$};
	 \node[node_style] (5) at ($\fscale*(-0.17, -0.17)$) {$5$};
	 \node[node_style] (6) at ($\fscale*(0.28, 0.28)$) {$6$};
	 \node[node_style] (7) at ($\fscale*(0.5, 0.28)$) {$7$};
	 \node[node_style] (8) at ($\fscale*(0.39, 0.06)$) {$8$};
	 \node[node_style] (9) at ($\fscale*(0.28, -0.17)$) {$9$};
	 \node[node_style] (10) at ($\fscale*(0.06, -0.28)$) {$10$};
	 \node[node_style] (0) at ($\fscale*(-0.5, 0.28)$) {$0$};

	 \draw (0) to  (2);
	 \draw (2) to  (5);
	 \draw (5) to  (10);
	 \draw (10) to  (9);
	 \draw (9) to  (8);
	 \draw (8) to  (6);
	 \draw (8) to  (7);
	 \draw (4) to  (5);
	 \draw (3) to  (5);
	 \draw (1) to  (2);
\end{tikzpicture}

\vspace{10pt}

\caption{Example of a tree and its postorder numbering ordered by decreasing depth. For example the children of $8$ are $6$ and $7$ and $\phi(9) = 6$ (the minimum descendant of $9$).}
        \label{fig:tree_postorder}
    \end{figure}

Consider $i <j$ where $j$ is an ancestor of $i$.
We denote by $x$ the parent of $i$ and by $y$ the parent of $y$.
As $x$ is the parent of $i$, $y$ is a cut vertex of $T[i,j]$ separating the set of vertices $X = \{i, \ldots, x\}$ from the rest of the tree.
As the numbering is in decreasing depth order, $T[i,x]$ is a star centered on $x$ (all the children of $x$ are leaves).
    Thus $T[X]$ is a star.
    
    If $d(y) = 2$ (see Figure~\ref{fig:bceovs_trees_algo_simple}), then by Lemma~\ref{lemma:bceovs_rec_simple}, where $z$ is the other neighbor of $y$, we have
    \[ t[i,j] =
    \begin{cases}
        1 + t[y+1, j] & \text{if } y \not=j \\
        1 + t[x+1, z] & \text{otherwise}.
    \end{cases}
    \]

    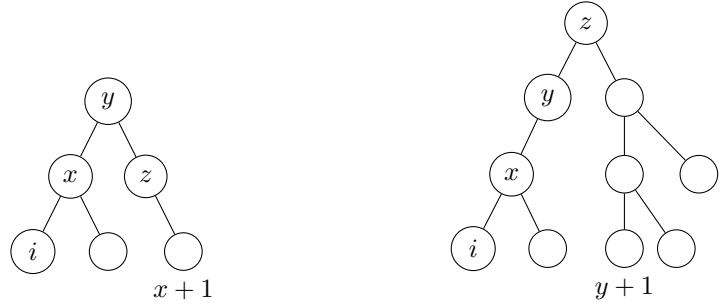
\begin{figure}[!h]
        \centering
        \begin{subfigure}[b]{0.49\textwidth}
            \centering
\begin{tikzpicture}
    [
        yscale=-1,
        node_style/.style={circle, draw, minimum size=0.5cm} 
    ]
    \def\fscale{2} 
	 \node[node_style] (1) at ($\fscale*(0, 0.5)$) {};
	 \node[node_style] (2) at ($\fscale*(-0.25, 0)$) {$x$};
	 \node[node_style] (0) at ($\fscale*(-0.5, 0.5)$) {$i$};
	 \node[node_style] (3) at ($\fscale*(0, -0.5)$) {$y$};
	 \node[node_style] (4) at ($\fscale*(0.25, 0)$) {$z$};
	 \node[node_style, label={below:$x+1$}] (5) at ($\fscale*(0.5, 0.5)$) {};

	 \draw (0) to  (2);
	 \draw (1) to  (2);
	 \draw (2) to  (3);
	 \draw (3) to  (4);
	 \draw (4) to  (5);
\end{tikzpicture}
            \caption{Case where $y=j$, so $y$ has no ancestor in $T[i,j]$ and $z$ is the other child of $y$.}
            \label{fig:bceovs_trees_algo_simple1}
        \end{subfigure}
        \hfill
        \begin{subfigure}[b]{0.49\textwidth}
            \centering
\begin{tikzpicture}
    [
        yscale=-1,
        node_style/.style={circle, draw, minimum size=0.5cm} 
    ]
    \def\fscale{3} 
	 \node[node_style] (1) at ($\fscale*(-0.17, 0.5)$) {};
	 \node[node_style, ] (2) at ($\fscale*(-0.33, 0.17)$) {$x$};
	 \node[node_style] (0) at ($\fscale*(-0.5, 0.5)$) {$i$};
	 \node[node_style] (3) at ($\fscale*(-0.17, -0.17)$) {$y$};
	 \node[node_style] (4) at ($\fscale*(0, -0.5)$) {$z$};
	 \node[node_style] (5) at ($\fscale*(0.17, -0.17)$) {};
	 \node[node_style] (6) at ($\fscale*(0.5, 0.17)$) {};
	 \node[node_style] (7) at ($\fscale*(0.17, 0.17)$) {};
	 \node[node_style, label={below:$y+1$}] (8) at ($\fscale*(0.17, 0.5)$) {};
	 \node[node_style] (9) at ($\fscale*(0.4, 0.5)$) {};

	 \draw (0) to  (2);
	 \draw (1) to  (2);
	 \draw (2) to  (3);
	 \draw (3) to  (4);
	 \draw (4) to  (5);
	 \draw (5) to  (6);
	 \draw (5) to  (7);
	 \draw (7) to  (8);
	 \draw (7) to  (9);
\end{tikzpicture}
            \caption{Case where $y <j$, so $y$ has only $x$ as a child and $z$ is the parent of $y$.}
            \label{fig:bceovs_trees_algo_simple2}
        \end{subfigure}
        \caption{Examples of cases where $d(y) = 2$. In these cases the edge $yz$ can be deleted.}
         \label{fig:bceovs_trees_algo_simple}
    \end{figure}

    then $bcevs(T[i,j]-X) = t[x+1, j]$.
    and $T-y-X$ is the disjoint union of $T[y+1,j]$ and the $T[\phi(k), k]$ where $k$ is a child of $y$.
    See Figure~\ref{figure:bceovs_trees_algo_multiple}.
    According to Lemma~\ref{lemma:bceovs_rec_multiple}, $bcevs(T[i,j]) = \min(1 + bcevs(T-A), d(y)-1 + bcevs(T-y-A))$.
    Thus, 
    \[
    t[i,j] = \min(1 + t[x+1,j], d(y)-1 + t[y+1,j] + \sum_{k \in children(y), k\not=x} t[\phi(k),k])
    \]

    \begin{figure}[!h] 
        \centering
\begin{tikzpicture}
    [
        yscale=-1,
        node_style/.style={circle, draw, minimum size=0.5cm} 
    ]
    \def\fscale{5} 
	 \node[node_style] (1) at ($\fscale*(-0.24, 0.39)$) {};
	 \node[node_style,] (2) at ($\fscale*(-0.37, 0.13)$) {$x$};
	 \node[node_style] (0) at ($\fscale*(-0.5, 0.39)$) {$i$};
	 \node[node_style] (3) at ($\fscale*(-0.24, -0.13)$) {$y$};
	 \node[node_style] (4) at ($\fscale*(0.02, -0.39)$) {};
	 \node[node_style] (5) at ($\fscale*(0.28, -0.13)$) {};
	 \node[node_style, label={below:$y+1$}] (7) at ($\fscale*(0.5, 0.13)$) {};
	 \node[node_style] (10) at ($\fscale*(-0.11, 0.13)$) {$k_1$};
	 \node[node_style] (8) at ($\fscale*(0.15, 0.13)$) {$k_2$};
	 \node[node_style, label={below:$\phi(k_1)$}] (9) at ($\fscale*(-0.11, 0.39)$) {};
	 \node[node_style] (12) at ($\fscale*(0.02, 0.39)$) {};
	 \node[node_style, label={below:$\phi(k_2)$}] (13) at ($\fscale*(0.15, 0.39)$) {};
	 \node[node_style] (14) at ($\fscale*(0.28, 0.39)$) {};

	 \draw (0) to  (2);
	 \draw (1) to  (2);
	 \draw (2) to  (3);
	 \draw (3) to  (4);
	 \draw (4) to  (5);
	 \draw (5) to  (7);
	 \draw (3) to  (10);
	 \draw (3) to  (8);
	 \draw (10) to  (9);
	 \draw (10) to  (12);
	 \draw (8) to  (13);
	 \draw (8) to  (14);
\end{tikzpicture}

\vspace{10pt}

\caption{Example where $d(y) \geq 3$ and $y \in A$. In this case either we delete the edges incident to $y$ except $xy$ and we use recursion, either we split $y$ to make $\{i, \ldots, x, y\}$ a cluster and we use recursion.}
        \label{figure:bceovs_trees_algo_multiple}
    \end{figure}
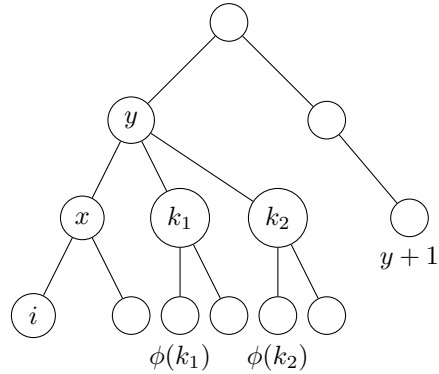

Therefore the resulting running time is in $O(n^2)$ by a recursive algorithm with memorization: at each time we compute the value of a $t[i,j]$ we store it in an array so that when we need this value later we get it from this store array if the value exists.
    We initialize our array with $t[i,i] = 0$ for every $i \in \{1, \ldots, n\}$.
\end{proof}

\section{Parameterized Complexity of \BCEOVS{}}

We show that the \textsc{Bicluster Editing with One-sided Vertex Splitting} problem is fixed-parameter tractable with respect to the total number of allowed operations. To do this, we present a kernelization algorithm (i.e., reduction procedure) for \BCEOVS{} that guarantees polynomial-size kernels.
In our reduction procedure we utilize an equivalent formulation of the \BCEOVS{}, again assuming 
 only the vertices of the set $B$ are allowed to split. The objective is to interpret the one-sided splitting from the set $B$ side as a partition of the set $A$, which boils down to finding a cover of the vertices of a bipartite graph $G=(A,B,E)$ such that the restriction to $A$ is a partition.

\begin{definition}
An $A$-partitioning cover $C$ of a bipartite graph $G=(A,B,E)$ is a set of subsets of $A \cup B$ covering the vertices of $G$ such that the restrictions of the subsets of $C$ to $A$ is a partition of $A$.
We define the cost of $C$ as follows:

    \begin{itemize}
        \item For every vertex $a \in A$, there exists a unique subset $X$ of $C$ such that $a \in X$ such that the cost of $a$ is defined as $|(B \cap X) \setminus N(a)| + |\overline{X} \cap N(a)|$ (where $\overline{X}$ denotes the complementary of $X$ in $A \cup B$).
        \item For every vertex $b \in B$, the cost of $b$ is defined as $j-1$ where $j$ is the number of subsets of $C$ containing $b$.
    \end{itemize}
    Finally, the cost of $C$ is defined as: $cost(C) = \sum_{a \in A} cost(a) + \sum_{b \in B} cost(b)$.
    
\end{definition}

The cost of a vertex $a \in A$ will correspond to the number of edited edges incident to $a$ and the cost of a vertex $b \in B$ will correspond to the number of times a vertex $b$ is split.
We will say that the edge $ab$, where $a \in$ will be \textit{deleted} if the unique subset $X$ of $C$ containing $a$ does not contain $b$.

\vspace{5pt}

\begin{figure}[!htb]
    \centering
\begin{tikzpicture}
    [
        yscale=-1,
        node_style/.style={circle, draw, inner sep=0.1cm} 
    ]
    \def\fscale{8} 
	 \node[node_style, dashed] (0) at ($\fscale*(-0.23, -0.01)$) {$a_2$};
	 \node[node_style] (1) at ($\fscale*(-0.06, -0.18)$) {$b_1$};
	 \node[node_style, dashed] (2) at ($\fscale*(0.11, -0.01)$) {$a_3$};
	 \node[node_style] (3) at ($\fscale*(-0.06, 0.15)$) {$b_2$};
	 \node[node_style] (4) at ($\fscale*(0.27, -0.18)$) {$b_3$};
	 \node[node_style, dashed] (5) at ($\fscale*(0.44, -0.01)$) {$a_4$};
	 \node[node_style] (6) at ($\fscale*(0.27, 0.15)$) {$b_4$};
	 \node[node_style, dashed] (7) at ($\fscale*(-0.39, -0.01)$) {$a_1$};

     \draw[rounded corners=5pt, color=gray] ($\fscale*(0.02, -0.26)$) -- ($\fscale*(-0.5, -0.26)$) -- ($\fscale*(-0.5, 0.26)$) -- ($\fscale*(0.02, 0.26)$) -- cycle;

    \draw[rounded corners=5pt, color=gray] ($\fscale*(-0.14, -0.22)$) -- ($\fscale*(-0.14, 0.2)$) -- ($\fscale*(0.18, 0.2)$) -- ($\fscale*(0.18, -0.22)$) -- cycle;

\draw[rounded corners=5pt, color=gray] ($\fscale*(0.22, -0.22)$) -- ($\fscale*(0.22, 0.2)$) -- ($\fscale*(0.5, 0.2)$) -- ($\fscale*(0.5, -0.22)$) -- cycle;

	 \draw (0) to  (1);
	 \draw (1) to  (2);
	 \draw (2) to  (3);
	 \draw (3) to  (0);
	 \draw (2) to  (4);
	 \draw (4) to  (5);
	 \draw (5) to  (6);
	 \draw (6) to  (2);
	 \draw (1) to  (7);
	 \draw (7) to  (3);
\end{tikzpicture}

\vspace{10pt}

\caption{Example of an $A$-partitioning cover of cost 4 where $A = \{a_1, a_2, a_3, a_4\}$. $cost(b_1) = cost(b_2) = 1$ because both vertices are in 2 subsets and $cost(a_3) = 2$ because the endvertices of the edges $a_3 b_3$ and $a_3 b_4$ are not in the same subset.}
\label{fig:cost_example}
\end{figure}
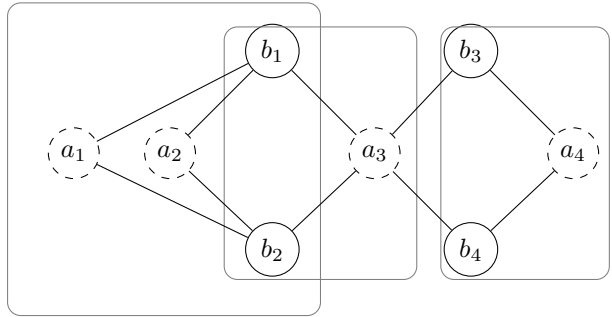

\begin{lemma} \label{lemma:cover_equivalence}
    A bipartite graph has an $A$-partitioning cover of cost at most $k$ if and only if there exists a sequence of length at most $k$ of edge editions and vertex splittings on $B$.
\end{lemma} 
\begin{proof}
    Suppose that $G$ has an $A$-partitioning cover of cost at most $k$.
    Consider the following sequence of operations:
    \begin{itemize}
        \item For every vertex $a \in A$, consider the unique subset $P$ of $C$ such that $a \in P$ (by definition of an $A$-partitioning cover): add all edges $ab$ where $b \in (B \cap P) \setminus N(a)$ and remove all edges $ab$ where $b \in (B \setminus P) \cap N(a)$.
        \item For every vertex $b \in B$, consider $j$ the number of subsets containing $b$.
        Then make $j-1$ copies of $b$ so that there are as much copies as there are subsets containing $b$.
        For every subset $P$ containing $b$, connect the copy $b_P$ to the $A$ vertices of $P$. 
    \end{itemize}
    This sequence of operations is of length the cost of $C$.
    Thus the length is at most $k$.
    Furthermore the resulting graph is a union of bicliques.

Assume that $G$ has a sequence of operations turning $G$ into a union of bicliques $B_1, \ldots, B_p$ called $G'$.
We define a cover $C$ of $G$ as follows:
For every biclique in $G'$, we define a subset $S$ of the vertices of $G$ consisting in the vertices of $A$ in the biclique and the vertices of $B$ which have a copy in the biclique. As no vertex of $A$ are split, then the restriction of $C$ to $A$ is a partition of $A$.
    Therefore it is an $A$-partitioning cover of $G$.
    Let us compute the cost of this one-sided partition.

    Let $a$ be a vertex of $A$.
    Let $K$ be all the vertices of $B$ such that $ab$ is deleted from $G$.
    Then the subset $P$ containing $a$ is such that $|\overline{P} \cap N(a)| = |K|$.
    Let $R$ be all the vertices of $B$ such that $ab$ is added to $G$.
    Then the subset $P$ containing $a$ is such that $|P \setminus N(a)| = |R|$.

    Let $b$ be a vertex of $B$.
    Let $j$ be the number of times $b$ is split.
    Then there is $j+1$ copies of $b$ in $G'$.
    Thus there is $j+1$ subsets containing $b$.
    Thus the $cost(b) = j+1-1 = j$.

    We conclude that the cost of the one-sided partition is the length of the sequence.
    Thus there exists an $A$-partitioning cover of cost at most $k$.
\end{proof}

\begin{definition}
    Let $G$ be a graph.
    A twin class is a maximal subset of the vertices of $G$ which have the same neighborhood.
\end{definition}

Observe that the twin classes of a graph partition the vertices of this graph.


\begin{lemma}[Twin adapted $A$-partitioning cover]
    \label{lemma:adapted_cover}
    Let $G=(A,B,E)$ be a bipartite graph.
    There exists an $A$-partitioning cover $C$ of $G$ of minimum cost such that for every twin class $T$ and every subset $X$ of $C$, then either $T \cap X = \emptyset$ or $T \subseteq X$.
\end{lemma}
\begin{proof}
    Let $C$ be a minimum $A$-partitioning cover of $G$.
    Consider a vertex $a \in A$ having the minimum cost among its twins.
    Let $P$ the subset of $C$ containing $a$.
    Let $a'$ be a twin of $a$ which is not in $P$.
    We move $a'$ to $P$.
    The new cost of $a'$ is $cost(a)$.
    By minimality of $a$, the cost of $X$ decreases.
    We can now assume that all the twins of $a$ are in the same subset.
    We repeat this operation for every vertex of the set $A$.

    Let $b$ be a vertex of $B$, whose subsets containing $b$ are noted $P_1, \ldots, P_r$, minimizing the quantity $r-1 + \sum_{i=1}^r P_i \cap \overline{N(b)} + \overline{\cup_{i=1}^r P_i} \cap N(b)$.
    Let $b'$ be a twin of $b$.
    Then move $b'$ so that it is contained in the same subsets as $b$.
    After this operation the cost of $X$ has decreased because the difference between the new cost of $X$ and the original cost of $X$ is:
    \[
    r-1 + \sum_{i=1}^r P_i \cap \overline{N(b)} + \overline{\cup_{i=1}^r P_i} \cap N(b) - (r'-1 + \sum_{i=1}^{r'} P'_i \cap \overline{N(b')} + \overline{\cup_{i=1}^{r'} P'_i} \cap N(b'))
    \]
    which is at most $0$ by minimality of $b$.
    We repeat this operation for every twin of $b$.
    By repeating these operations for every vertex of $B$, we have found an $A$-partitioning cover of minimum cost which satisfies the property.

\end{proof}




    

\begin{lemma}
    \label{lemma:number_A_twin_classes}
    Let $G=(A,B,E)$ be a connected bipartite graph with $k$ twin classes in $A$.
    Then the cost of any $A$-partitioning cover is at least $\sqrt{k}-1$.
\end{lemma}
\begin{proof}
    By Lemma~\ref{lemma:adapted_cover}, consider an $A$-partitioning cover $C$ of $G$ of minimum cost which is adapted to the twin classes of $G$ .
    We denote by $c$ the number of subsets of $C$.


    Construct a graph $G'$ with one vertex for every subset of $C$ and one vertex for every vertex of $B$.
    Thus $G'$ has $n' = c + |B|$ vertices.
    Connect in $G'$ the vertex of a subset $X$ of $C$ to a vertex $b$ if $b \in N[X]$.
    As $G$ is connected, then $G'$ is also connected.
    Thus the number $m'$ of edges of $G'$ is at least $ n'-1 = c+ |B|-1$.

    Let us now prove that $cost(C)  \geq \sum_{b \in B} (d(b)-1)$ where $d(b)$ denotes the degree of $b$ in $G'$.
    For any vertex $b$ of $B$, we denote by $del(b)$ the number of edges incident to $b$ that will be deleted: it is the number of vertices $a \in A$ such that $ab \in E$ and such that the unique subset $X$ of $C$ containing $a$ does not contain $b$.
    
    Let $b$ be a $B$ vertex.
    We denote by $s$ the cost of $b$.
    Then $b$ is in $s+1$ subsets of $C$.
    Let $X_1, \ldots, X_d$ be the subsets of $C$ such that $b \in \cap_{i=1}^d N[X_i]$.
    Remark that $s+1 \leq d$. 
    If $s = d-1$, then $cost(b) + del(b) \geq d-1$.
    Otherwise $s \leq d-2$.
    There exists at least $d-(s+1)$ subsets of $X_1, \ldots, X_d$ in which $b$ is not contained.
    Therefore, for each subset $X \in C$ such that $b \not\in X$ but $b \in N[X]$, there exists a vertex $a \in X$ such that $b \in N(a)$.
    Thus $cost(b) + del(b) \geq s+ (d - s -1) = d-1$.

    Remark that $cost(C) =  \sum_{b \in B} cost(b) + \sum_{a \in A} cost(a) \geq \sum_{b \in B} (cost(b) + del(b))$.
    Because of the previous paragraph, we conclude that we have the inequality $cost(C) \geq \sum_{b\in B} (d(b)-1)$.

    As $G'$ is bipartite, the number $m'$ of edges of $G'$ equals $\sum_{b \in B} d(b)$.
    Thus,  $cost(C) \geq  m' - |B|$.
    As $m' \geq c + |B| -1$, then the cost of $C$ is at least $c-1$.

    Let us now make a disjunction of cases.
    If $c \geq \sqrt{k}$, then $cost(C) \geq \sqrt{k}-1$ because of the previous inequality.
    
    Otherwise, assume $c \leq \sqrt{k}$.
    As $C$ is a twin adapted $A$-partitioning cover of $G$, $k = \sum_{i=1}^c ATC(X_i)$ where $ATC(X_i)$ is the number of $A$ twin classes contained in $X_i$.
    Thus there exists a subset $X$ of $C$ containing at least $\sqrt{k}$ $A$ twin classes because otherwise $ATC(X_i) < \sqrt{k}$ for every $i$ and thus $\sum_{i=1}^c ATC(X_i) < c \sqrt{k} \leq k$, a contradiction.

    Let $x$ be the number of $A$ twin classes contained in $X$.
    Let us show that $\sum_{a \in A \cap X} cost(a) \geq x-1$.
    
    If every $A$ vertex of $X$ has a non zero cost, then the sum of the costs of the $A$ vertices of $X$ is at least $|X| \geq x$.
    Otherwise, suppose that there exists a vertex $a \in A$ in $X$ of cost 0.
    Then $N(a)$ must be in $X$.
    For every $a' \in A$ such that $a'$ is a vertex of $X$ of a different twin class than $a$, $N(a') \not= N(a)$.
    Let $b \in N(a') \Delta N(a)$.
    If $b \in N(a') \setminus N(a)$ then $b \not\in X$, because otherwise the cost of $a$ would be non zero.
    Therefore the edge $a'b$ will be deleted, and thus the cost of $a'$ is at least 1.
    If $b \in N(a) \setminus N(a')$.
    As $b \in X$, the edge $ab$ will be added and then the cost of $a'$ is at least 1.
    We deduce that the sum of the costs of the $A$ vertices in $X$ is at least $x-1$, because except for the twin class of $a$, every other vertex is of cost at least 1.

    Hence, in any case, the sum of the costs of the $A$ vertices in $X$ is at least $x-1$.
    As $x \geq \sqrt{k}$, we deduce that the cost of $C$ is at least $\sqrt{k} -1$.
    By minimality of $C$, we conclude that every $A$ partitioning cover has a cost of at least $\sqrt{k}-1$.
\end{proof}

\begin{lemma}[Reduction]
    \label{lemma:types_size}
    Let $G$ be a bipartite graph given along with an integer $k$.
    Consider the twin classes $T_1, \ldots, T_p$ of $G$.
    For every $i \in [p]$, we consider a subset $T_i'$ of $T_i$ of size $k+1$ if $|T_i| \geq k+1$, otherwise we set $T_i'$ to $T_i$.
   Then $G$ has an $A$-partitioning cover of cost at most $k$ if and only $G[\cup T_i']$ has an $A$-partitioning cover of cost at most $k$.
    

\end{lemma}
\begin{proof}
    We denote by $G'$ the graph $G[\cup T_i']$.
    
    Assume that there exists an $A$-partitioning cover $C$ of the vertices of $G$ of cost at most $k$.
    We consider the cover $C'$ induced by the cover $C$ on the vertices of $G'$ (which is an induced subgraph of $G$).
    The costs of the vertices can only be decreased by deleting vertices.
    Therefore the cover of $C'$ must have a smaller cost than $C$.

    Assume that $G'$ has an $A$-partitioning cover of cost at most $k$.
    Thus there exists an $A$-partitioning cover $C'$ of $G'$ of cost at most $k$ which is adapted to the twin classes.
    We define an $A$-partitioning cover $C$ of $G$ by extending the subsets of $C'$ to the twin classes of $G$.

We shall prove that $C$ and $C'$ have the same cost.
    Let $x$ be a vertex of a twin class which has been reduced.
    Then $x$ must have $k$ twins.
    As the cover is adapted to the twin classes, then the twins of $x$ are contained in the same subsets.
    Therefore they must all have the same cost.
    It is therefore impossible for the cost of $x$ to be non zero, otherwise the cost of $x$ and all of its twins would be larger than $k+1$.
    This implies that $cost(x) = 0$ in $C'$.
    Therefore the cost of $x$ is also $0$ in $C$.
    Hence $C$ and $C'$ have the same cost.
\end{proof}

\begin{lemma}
    \label{lemma:critical_biclique_size}
    Let $G=(A,B,E)$ be a connected bipartite graph such that all twin classes are of size at most $k$ and having a twin adapted $A$-partitioning cover $C$ of cost at most $k$.
    Then there are at most $4k^4$ twin classes in $B$.
\end{lemma}
\begin{proof}
    Consider a twin class $T$ of $A$, and a vertex $a \in T$.
    Then $|T| \leq k$.
    Let $t$ be the number of twin classes that have a vertex in $N(T)$.
    Let us prove that $t \leq 2k+1$.
    
    Assume that $t \geq 2k+2$.
    Then there exists $B$ vertices $y_1, y_2, \ldots, y_{2k+2}$ in $N(T)$ of different twin classes.
    Consider two different indices $i$ and $j$.
    Let us prove that there exists a vertex $a' \in A$ such that the edge $a' y_i$ or $a' y_j$ is either deleted or added.
    
    As $y_i$ and $y_j$ have a different neighborhood, there exists $a' \in A$ such that $a'$ is adjacent to $y_i$ but not to $y_j$ (we can swap $i$ and $j$ if $a'$ is not adjacent to $y_i$ but to $y_j$).
    The vertex $a'$ is not in $T$ (otherwise $a'$ would be adjacent to $y_i$ and $y_j$).
    Thus $a' \not= a$.

    If $a$ and $a'$ are in a same subset $Y$ of $C$, and
$y_j$ is also in this set, then the edge $a' y_j$ will be added.
    Otherwise $y_j \not\in Y$ and in this case the edge $ay_j$ will be deleted.
    Otherwise $a$ and $a'$ are in different subsets $X$ and $X'$ of $C$.
    If $y_i$ is not in both $X$ and $X'$, then the edge $a y_i$ or $ay_j$ is deleted.
    Otherwise $y_i$ is in both $X$ and $X'$.
    Thus the cost of $y_i$ is at least 1.

    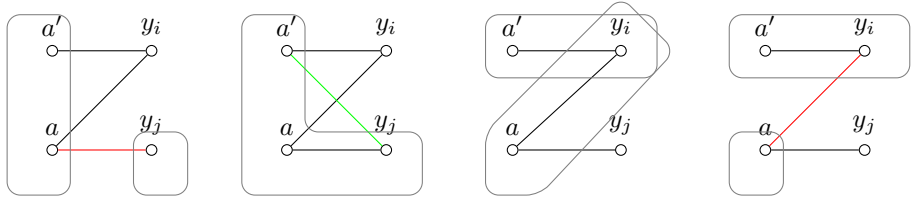
\begin{figure}[h]
        \centering
\begin{tikzpicture}
    [
        yscale=-1,
        node_style/.style={circle, draw, inner sep=0.05cm} 
    ]
    \def\fscale{12} 
	 \node[node_style, label={above:$a'$}] (0) at ($\fscale*(-0.45, -0.05)$) {};
	 \node[node_style, label={above:$y_i$}] (1) at ($\fscale*(-0.34, -0.05)$) {};
	 \node[node_style, label={above:$a$}] (2) at ($\fscale*(-0.45, 0.06)$) {};
	 \node[node_style, label={above:$y_j$}] (3) at ($\fscale*(-0.34, 0.06)$) {};
	 \node[node_style, label={above:$a'$}] (4) at ($\fscale*(0.06, -0.05)$) {};
	 \node[node_style, label={above:$y_i$}] (5) at ($\fscale*(0.18, -0.05)$) {};
	 \node[node_style, label={above:$a$}] (6) at ($\fscale*(0.06, 0.06)$) {};
	 \node[node_style, label={above:$y_j$}] (7) at ($\fscale*(0.18, 0.06)$) {};
	 \coordinate (8) at ($\fscale*(-0.5, -0.09)$) {};
     
	\coordinate (9) at ($\fscale*(-0.5, 0.11)$) {};
	\coordinate (10) at ($\fscale*(-0.43, 0.11)$) {};
	\coordinate (11) at ($\fscale*(-0.43, -0.09)$) {};
	\coordinate (12) at ($\fscale*(0.03, -0.09)$) {};
	\coordinate (13) at ($\fscale*(0.03, -0.02)$) {};
	\coordinate (14) at ($\fscale*(0.22, -0.02)$) {};
	\coordinate (15) at ($\fscale*(0.22, -0.09)$) {};
	\coordinate (16) at ($\fscale*(0.03, 0.04)$) {};
	\coordinate (17) at ($\fscale*(0.03, 0.11)$) {};
	\coordinate (18) at ($\fscale*(0.09, 0.11)$) {};
	\coordinate (20) at ($\fscale*(-0.3, 0.04)$) {};
	\coordinate (21) at ($\fscale*(-0.3, 0.11)$) {};
	\coordinate (22) at ($\fscale*(-0.36, 0.04)$) {};
	\coordinate (23) at ($\fscale*(-0.36, 0.11)$) {};
	 \node[node_style, label={above:$a'$}] (24) at ($\fscale*(-0.19, -0.05)$) {};
	 \node[node_style, label={above:$y_i$}] (25) at ($\fscale*(-0.08, -0.05)$) {};
	 \node[node_style, label={above:$a$}] (26) at ($\fscale*(-0.19, 0.06)$) {};
	 \node[node_style, label={above:$y_j$}] (27) at ($\fscale*(-0.08, 0.06)$) {};
	\coordinate (28) at ($\fscale*(-0.24, -0.09)$) {};
	\coordinate (29) at ($\fscale*(-0.24, 0.11)$) {};
	\coordinate (31) at ($\fscale*(-0.17, -0.09)$) {};
	\coordinate (32) at ($\fscale*(-0.04, 0.04)$) {};
	\coordinate (33) at ($\fscale*(-0.04, 0.11)$) {};
	\coordinate (34) at ($\fscale*(-0.17, 0.04)$) {};
	 \node[node_style, label={above:$a'$}] (30) at ($\fscale*(0.34, -0.05)$) {};
	 \node[node_style, label={above:$y_i$}] (35) at ($\fscale*(0.45, -0.05)$) {};
	 \node[node_style, label={above:$a$}] (36) at ($\fscale*(0.34, 0.06)$) {};
	 \node[node_style, label={above:$y_j$}] (37) at ($\fscale*(0.45, 0.06)$) {};
	\coordinate (38) at ($\fscale*(0.3, -0.09)$) {};
	\coordinate (39) at ($\fscale*(0.3, -0.02)$) {};
	\coordinate (40) at ($\fscale*(0.5, -0.02)$) {};
	\coordinate (41) at ($\fscale*(0.5, -0.09)$) {};
	\coordinate (42) at ($\fscale*(0.3, 0.04)$) {};
	\coordinate (43) at ($\fscale*(0.3, 0.11)$) {};
	\coordinate (44) at ($\fscale*(0.36, 0.11)$) {};
	\coordinate (45) at ($\fscale*(0.36, 0.04)$) {};
	\coordinate (19) at ($\fscale*(0.18, -0.11)$) {};
	\coordinate (46) at ($\fscale*(0.24, -0.05)$) {};

	 \draw (0) to  (1);
	 \draw (2) to  (1);
	 \draw[color=red] (2) to  (3);
	 \draw (4) to  (5);
	 \draw (6) to  (7);

\draw (24) to  (25);
\draw (26) to  (27);
\draw (26) to  (34);
\draw (34) to  (25);
\draw (30) to  (35);
\draw (36) to  (37);
\draw (6) to  (5);

\draw[color=green] (24) to  (27);

\draw[color=red] (36) to  (45);
\draw[color=red] (45) to  (35);

\draw[rounded corners=5pt, color=gray] (10) -- (11) -- (8) -- (9) -- cycle;
\draw[rounded corners=5pt, color=gray] (12) -- (13) -- (14) -- (15) -- cycle;
\draw[rounded corners=5pt, color=gray] (20) -- (22) -- (23) -- (21) -- cycle;

\draw[rounded corners=5pt, color=gray] (16) -- (17) -- (18)  -- (46) -- (19) --  cycle;

\draw[rounded corners=5pt, color=gray] (38) -- (39) -- (40)  -- (41) --  cycle;

\draw[rounded corners=5pt, color=gray] (32) -- (34) -- (31)  -- (28) -- (29) -- (33) --  cycle;

\draw[rounded corners=5pt, color=gray] (42) -- (43) -- (44)  -- (45) --  cycle;

\end{tikzpicture}

\vspace{10pt}
        \caption{On the left, $a$ and $a'$ are in the same set of the cover and $y_j$ is not in the same set.
        In the second case, $a$, $a'$ and $y_j$ are in the same set, therefore $a'y_j$ will be added.
        In the third figure, $a$ and $a'$ are in different sets both containing $y_i$; therefore $y_i$ is split at least once.
        In the fourth figure, $a$ and $a'$ are not in the same set and $y_i$ is not in the set of $a$, therefore the edge $ay_i$ is added. }
        \label{fig:b_twin_classes_2}
    \end{figure}
    
    We deduce that in every case, either an edge adjacent to $y_i$ or $y_j$ is deleted or added, or $y_i$ has a cost of at least $1$.
    By summing up all these cases for every pair $y_{2i-1} y_{2i}$ for every $i \in \{1, \ldots, k+1\}$, we conclude that the cost of $C$ is at least $k+1$.
    A contradiction.
    Hence $t \leq 2k+1$.
    
    By our hypothesis, the twin classes are of size at most $k$.
    Therefore $N(T)$ is of size at most $k(2k+1) \leq 2k^2$.
    Because of Lemma~\ref{lemma:number_A_twin_classes}, there are at most $(k+2)^2-1$ twin classes in $A$ (otherwise the cost of $C$ would be at least $k+1$).
    As the graph is connected, each twin class in $B$ is connected to at least one twin class in $A$.
    Thus there are at most $((k+2)^2-1 ) \cdot 2k^2 \leq 4k^4$ twin classes in $B$.
\end{proof}

\begin{theorem}
    \textsc{Bicluster Editing with One-sided Vertex Splitting} has an $O(k^5)$ kernel.
\end{theorem}
\begin{proof}
    Computing the twin classes can be done in $O(n^2)$ time.
    If there are more than $(k+2)^2-1$ twin classes in $A$ or more than $4(k+1)^4$ twin classes in $B$, then we set $G'$ to a path of length $4(k+1)$.
    This graph needs at least $k+1$ editing operations to be turned into a disjoint union of bicliques because there are $k+1$ disjoint geodesics of length $3$ each.
    Now because of Lemmas~\ref{lemma:number_A_twin_classes}, \ref{lemma:types_size} and \ref{lemma:critical_biclique_size}, we must have a no-instance.

    Otherwise there are at most $(k+2)^2-1$ twin classes in $A$ and $4(k+1)^4$ twin classes in $B$.
    We set $G'$ to be the graph constructed by Lemma~\ref{lemma:types_size}.
    Then, in this case, $G'$ has a total of at most $2k^2 + 4(k+1)^4$ twin classes which are all of size at most $k+1$ (each).
    Therefore $G'$ is of size at most $7k^5$, and $G'$ has a cover of cost at most $k$ if and only if $G$ has a cover of cost at most $k$.

    We conclude that, in all cases, we have constructed a reduced problem instance (the graph $G'$) of size at most $7k^5$, in polynomial time, such that $G'$ has a cover of cost at most $k$ if and only if $G$ has a cover of cost at most $k$.
\end{proof}


\section{Hardness of approximation}

Our objective in this section is to reduce MAX 3-SAT(4) to \BCEVS{} and \BCEOVS{}. 
MAX 3-SAT(4) is a variant of MAX 3-SAT where each variable appears at most four times in a formula $\phi$.

We also add the following constraint: when a variable $v$ appears exactly two times positively and two times negatively, we suppose that the list of clauses in which $v$ appears, $C(v)_1, C(v)_2,C(v)_3,C(v)_4$, is ordered so that $v$ appears positively in $C(v)_1$ and $C(v)_3$ and negatively in $C(v)_2$ and $C(v)_4$.  
This constraint is added to ensure that each unsatisfied clause in $\phi$ causes an additional split in the construction.
In fact, we can observe that if the formula $\phi$ cannot be satisfied, then we can use an ``extra'' split in each clause gadget to obtain a solution. 
However, the inverse does not necessarily hold if there is a variable $v$ that occurs two times positively and two times negatively.
Indeed by using $8+1$ splits in the variable cycle, we may be able to satisfy the four clauses where $v$ occurs.

\begin{definition}[Linear reduction \cite{PapadimitriouY91}]
    Let $A$ and $B$ be optimization problems having each one its own cost function defined on solutions of their instances.
    We say that $A$ has a linear reduction to $B$ if there exists two polynomial time algorithms $f$ and $g$ and two positive numbers $\alpha$ and $\beta$ such that for any instance $I$ of $A$
    \begin{itemize}
        \item $f(I)$ is an instance of $B$.
        \item $OPT_B(f(I)) \leq \alpha OPT_A(I)$ .
        \item For any solution $S'$ of $f(I)$, $g(S')$ is a solution of $I$ and   $|cost_A(g(S'))- OPT_A(I)| \leq \beta |cost_B(S') - OPT_B(f(I))|$.
    \end{itemize}
\end{definition}

\begin{theorem}\label{th:inapprox}
The problems \BCEVS{} and \BCEOVS{} are {\sc APX}-hard.
\end{theorem}
\begin{proof}
First, note that it is $\NP$-hard to approximate {\sc MAX 3-SAT(4)} to any factor $\epsilon_4 \leq 1.00052$ 
~\cite{karpinski}. 

Consider an instance $\phi$ of {\sc MAX 3-SAT(4)}.
We denote by $M$ the number of clauses of $\phi$.
According to \cite{Hastad97}, there exists an assignment which satisfies at least $\frac{7}{8}$ of the clauses.
By denoting $OPT(\phi)$ the maximum number of clauses that can be satisfied by an assignment, we have
  \begin{equation}
    OPT(\phi) \geq \frac{7M}{8}.\label{eq:7m by 8}
  \end{equation}

\vspace{5pt}
Let us show that Construction~\ref{construction:reduction} is a linear reduction.
Let $f$ be a function mapping an instance $\phi$ of {\sc MAX 3-SAT(4)} into the graph $G_\phi$ obtained by Construction~\ref{construction:reduction}.
Let $\phi$ be such an instance.
We denote by $G$ the graph obtained from Construction~\ref{construction:reduction}.

Let $X$ be a sequence of edge editions and splits turning $G$ into a disjoint union of bicliques such that $X$ deletes exactly 8 edges per variable gadget and deletes two edges in each clause gadget.

We denote by $cost(X)$ the length of the sequence.
Let $g$ be the function that transforms $X$ into a boolean assignment as constructed in the proof of Theorem~\ref{theorem:npc}: each variable $v$ is set to true if $X$ deletes all the edges $v_{2+2k} v_{3+3k}$ for every $k$, and false, otherwise.

We denote by $cost(g(X))$ the number of satisfied clauses.
If a clause is not satisfied, then its corresponding clause gadget contains three splits and two otherwise.

Delete the edges of the clauses and delete the edges $v_{3i} v_{3i+1}$ for every variable $v$ for every $i$.
The graph obtained is a disjoint union of bicliques.
As we use $3M$ operations for the clauses and $6M$ operations for the variables, we have
\begin{equation}
  OPT(G_\phi) \leq 6M + 3M \leq 9M  \stackrel{\eqref{eq:7m by 8}}{\leq} 9\cdot \frac{8}{7} OPT(\phi) \leq \frac{72}{7} OPT(\phi) \label{eq:I and phi}.
\end{equation}

Let $\sigma$ be an assignment of $\phi$ maximizing the number of satisfied clauses of $\phi$.
Let us define a sequence of operations on $G_\phi$.
For every true (resp. false) variable $v$, we delete the edges $v_{3i} v_{3i}$ (resp. $v_{3i+2} v_{3i+3}$) for every $i$.
For every clause $c$ with variables $u,v$ and $w$, if $c$ is satisfied we can suppose that $v$ is satisfying $c$, we delete the edges $c u_c$ and $c w_c$ (like in the proof of Theorem~\ref{theorem:npc}).
Otherwise we delete the three edges $c u_c$, $c v_c$ and $c w_c$.
We denote by $SC$ the set of satisfied clauses of $\phi$ by $\sigma$ and $UC$ the set of unsatisfied clauses of $\phi$ by $\sigma$.
This sequence is of length $\sum_{v \in V} 2 d(v) + \sum_{c \in SC} 2 + \sum_{c \in UC} 3 = 6M + 2M + k$ where $k$ is the number of unsatisfied clauses.
Furthermore this sequence of operations turns the graph $G_\phi$ into a disjoint union of bicliques.
We deduce that $OPT(G_\phi) \leq 8M + M-OPT(\phi)$.
Thus $OPT(\phi) + OPT(G_\phi) \leq 8M + M$.

Let $X$ be a sequence of operations on $G_\phi$ turning this graph into a disjoint union of bicliques.
Let us define the assignment $\sigma$ as follows.

Consider a variable $v$.
According to Lemma~\ref{lemma:cycle}, the variable cycle of $v$ needs at least $2d(v) = 8$ operations on its edges and its vertices.

If the variable cycle is using exactly $8$ operations, then because of Lemma~\ref{lemma:cycle}, there are three cases.
If the edges $v_{1+3i} v_{2+3i}$ are deleted for every $i$, then we assign $v$ to positive and to negative otherwise.

If the variable cycle is using at least $8+2$ operations, then we replace these operations by deleting the edges $v_{3i} v_{3i+1}$ for every $i$ and by splitting the vertices $v_0$ and $v_{12}$ so that it disconnects the variable cycle and the two positive clauses connected to $v$.
We assign $v$ to negative.

If the variable cycle is using exactly $8+1$, let us show that it is not possible to solve the geodesics $c_0, v_0, v_1, v_2$, $c_0, v_0,v_{23}, v_{22}$, $c_1, v_8, v_7, v_6$, $c_1, v_8, v_9, v_{10}$, $c_2, v_{12}, v_{11}, v_{10}$, $(c_2, v_{12}, v_{13}, v_{14})$, $(c_3, v_{20}, v_{19}, v_{18})$ and $(c_4, v_{20}, v_{21}, v_{22})$.
Then at least one of the edge $c_0 v_0$, $c_1 v_8$, $c_2 v_{12}$ and $c_3 v_{20}$ should be deleted or one of the vertex $v_0$, $v_8$, $v_{12}$ and $v_{20}$ should be split.
Otherwise we can find $10$ geodesics of length 4 which are two by two independent.
Therefore only 3 clauses geodesics can be resolved. 
If the two clauses where $v$ appears positively, then we assign $v$ to positive.
Otherwise we assign $v$ negatively.

In any case we remark that we need to use at least one operation for every unsatisfied clause.
We denote by $UC$ the number of unsatisfied clauses by the assignment $\sigma$.
We deduce that $UC \leq cost(X) - 8M$.
Thus
\begin{align*}
OPT(\phi) + OPT(G_\phi) &\leq 9M  \leq cost(X) +M - UC \\
OPT(\phi) + OPT(G_\phi) &\leq 9M  \leq cost(X) + cost(g(X)) \\
OPT(\phi) - cost(g(X)) &\leq cost(X) - OPT(G_\phi)
\end{align*}

Thus, we have constructed a linear reduction with $\alpha=\frac{72}{7}, \beta=1$. 
As {\sc MAX 3-SAT(4)} is {\sc APX}-hard, we deduce that \BCEVS{} and \BCEOVS{} are \APX-hard.
\end{proof}

\section*{Acknowledgements}

This research project was supported by the Lebanese American University under the President’s Intramural Research Fund PIRF0056.

\section{Concluding Remarks}

This paper introduces the \textsc{Bicluster Editing with Vertex Splitting} problem (\BCEVS{}) and the \textsc{Bicluster Editing with One-Sided Vertex Splitting} problem (\BCEOVS{}).
Both \BCEVS{} and \BCEOVS{} have been shown to be \NP-complete even when restricted to  bipartite planar graphs of degree at most three. We also proved the two problems are APX-hard. On the positive side, a fixed-parameter algorithm was presented for \BCEOVS{} and the two problems are proved to be solvable in polynomial-time on trees. 
This latter result might seem to be of limited importance, but it suggests that the problems might be fixed-parameter tractable when parameterized by the treewidth of the input graphs, which is hereby posed as an open problem.

Future work may focus on proving the \NP-completeness of (either of) the two problems on other classes of bipartite graphs.
The APX-hardness of the problems leads to the question whether finding a polynomial time constant-factor approximation is possible.
We have further shown that \BCEOVS{} is \FPT with respect to the number $k$ of operations. 
Thiswith was the result of presenting a kernelization algorithm with a kernel bound in \(O(k^5)\).
More recently, and capitalizing on the work presentd in our  conference version, Bentert et al. obtained a quadratic size kernel bound and presented a fixed-parameter algorithm that runs in $\bigo(k^{11k}+n+m)$. It would be interesting to look for a better fixed-parameter algorithm. Most importantly, would it be possible (modulo ETH, for example) to obtain a $\bigo^*
(c^k)$ algorithm? The same questions would be also interesting in the case of 
\BCEVS{}.
Another interesting future work would be to consider other auxiliary parameters like, in particular, twin-width; as well as local parameters (as in \cite{abu2017complexity,Heggernes10,komusiewicz2012cluster}) such as a bound on the number of times a vertex can be split. This latter bound is motivated by real applications where a data element cannot belong to an arbitrary large number of (bi)clusters.

\end{document}